\documentclass[10pt]{llncs}

\usepackage[dvips]{graphicx}
\usepackage{times}
\usepackage{color}
\usepackage{listings}
\usepackage{amsmath}
\usepackage{amssymb}
\usepackage{stmaryrd}
\usepackage[hyphens,spaces]{url}
\usepackage[bookmarks=false]{hyperref}
\usepackage{algorithmicx}
\usepackage{algorithm}
\usepackage{xparse}
\usepackage{wrapfig}
\usepackage{etoolbox}
\usepackage{cite}
\usepackage{xcolor}
\usepackage{textcomp}
\usepackage{lmodern}

\newtoggle{tech_report}

\toggletrue{tech_report}

\iftoggle{tech_report}{

}{
}

\newcommand{\exampletext}[1]{\lstinline{#1}}
\def\buyer{{\exampletext{buyer}}}
\def\seller{{\exampletext{seller}}}

\def\bank{{\exampletext{bank}}}

\def\pay{{\exampletext{pay}}}

\def\AdaptRule{Up}
\def\NoAdaptRule{NoUp}
\def\NoAdaptLabel{\func{no-up}}

\def\net{\mathcal{N}}
\def\event{\func{events}}
\def\devent{\func{events}}

\def\rules{\mathbf{I}}
\def\APOC{\mbox{$\func{DPOC}$}}
\def\AIOC{\mbox{$\func{DIOC}$}}

\def \mid {\; | \;}

\newcommand{\one}{\mathbf{1}}
\newcommand{\zero}{\mathbf{0}}
\newcommand{\tick}{\surd}

\newcommand{\ruleSpace}{\vspace{0em}}
\def \mathax #1#2{\ruleSpace\begin{array}{l}{\mbox{\scriptsize \ruleName{#1}} } \\[-.1cm] #2
\end{array}}
\def \mathrule #1#2#3{\ruleSpace\begin{array}{l}%
    {\mbox{\scriptsize \ruleName{#1}}}
    \\ \bigfract{#2}{#3}
\end{array}}

\newcommand{\bigfract}[2]{\frac{^{\textstyle #1}}{_{\textstyle #2}}}

\newcommand{\state}{\Sigma}

\newcommand{\substLabel}[3]{[#1 / #2,#3]}
\newcommand{\substAPOC}[2]{[#1 / #2]}
\newcommand{\comm}[5]{#1 : #2(#3) \rightarrow #4(#5)}
\newcommand{\sign}[3]{#1 : #2 \rightarrow #3}
\newcommand{\commLabel}[5]{#1 : #2(#3) \rightarrow #4(#5)}
\newcommand{\co}[2]{\overline{#1} \at {#2}}

\newcommand{\cout}[3]{#1 : #2\; \code{to}\; {#3}}
\newcommand{\coutLabel}[3]{\overline{#1}\langle#2\rangle \at {#3}}
\newcommand{\ci}[2]{#1 \at #2}
\newcommand{\cinp}[3]{#1 : #2 \;\code{from}\; #3}
\newcommand{\cinpLabel}[4]{#1(#2 \leftarrow #3) \at #4 }

\DeclareMathOperator{\upd}{\func{upd}}
\DeclareMathOperator{\prop}{\func{prop}}
\DeclareMathOperator{\ssim}{\func{sim}}

\DeclareMathOperator{\gram}{::=}

\DeclareMathOperator{\proj}{\func{proj}}
\DeclareMathOperator{\freshKey}{\func{freshIndex}}

\newcommand{\sset}[2]{\left\{~#1  \left |
                    \begin{array}{l}#2\end{array} \right.     \right\}}
\newcommand{\ruleName}[1]{[{\sc #1}]}
\newcommand{\rulebox}[1]{\fcolorbox{gray}{white}{#1}}
\newcommand{\arule}[3]{#3}
\newcommand{\cond}{\mathcal{C}}

\NewDocumentCommand\scope{mmmmg}{\code{scope} \ \at #2 \ \{ #3 \} %
}
\NewDocumentCommand\pscope{mmmmmg}{#1: \code{scope} \ \at #3 \ \{ #4 \} %
\IfNoValueTF{#6}{}{\ \code{roles} \ \{ #6 \}}}
\newcommand{\psscope}[4]{#1: \code{scope} \ \at #3 \ \{ #4 \}}
\def\grass#1{[\![#1]\!]}
\def\ev{\varepsilon}
\def\leqapoc{\leq_{\APOC}}
\def\leqaioc{\leq_{\AIOC}}

\def\prod{\Pi}

\newcommand{\code}[1]{\mbox{\normalfont{\footnotesize{\textbf{\texttt{#1}}}}}}
\newcommand{\func}[1]{\mbox{\normalfont{\footnotesize{\textsf{#1}}}}}

\newcommand{\seqOp}{;}
\newcommand{\parOpI}{|}
\newcommand{\parOpP}{|}
\newcommand{\freshops}[2]{#1 \cdot #2} 
\newcommand{\assign}[3]{#1 \at #2 = #3 }
\newcommand{\at}{\code{@}}
\newcommand{\pair}[2]{#1 \rightarrow #2}

\newcommand{\adapt}[4]{\code{scope} \; \at #4 \; \{ #1 \} %
}

\newcommand{\ifthen}[3]{\code{if} \; #1  \; \{ #2 \} \; \code{else} \; \{ #3 
\}}
\newcommand{\ifthenKey}[4]{#4 : \code{if} \; #1  \; \{ #2 \} \; \code{else} \; \{ #3 
\}}
\newcommand{\while}[2]{\code{while} \; #1 \;  \{ #2 \}}
\newcommand{\whileKey}[3]{#3 : \code{while} \; #1 \;  \{ #2 \}}

\newcommand{\arro}[1]{\xrightarrow[]{#1}}

\newcommand{\tuple}[1]{\left\langle #1\right\rangle}

\DeclareMathOperator{\transI}{\func{transI}}

\DeclareMathOperator{\transF}{\func{transF}}
\DeclareMathOperator{\operations}{\func{sig}}

\DeclareMathOperator{\roles}{\func{roles}}

\DeclareMathOperator{\bisim}{\sim}

\newcommand{\role}[2]{(#1)_{#2}}

\definecolor{color:keyword}{rgb}{0.53,0.05,0.05}
\definecolor{color:comment}{rgb}{0.25,0.37,0.75}
\definecolor{color:string}{rgb}{0.87,0.0,0.0}
\usepackage{bold-extra}

\lstdefinelanguage{MyChor}{
morekeywords={
  while, if, else, scope, true, false, not, getInput, and, prop, from, to, via, on, do, and, 
roles, null
},
sensitive=true,
morecomment=[l]{//},
morecomment=[s]{/*}{*/},
morestring=[b]",
otherkeywords={;,|,:,@,!}
}

\newcommand{\setMyChor}{
\lstset{
  language=MyChor,
  basicstyle=\ttfamily\footnotesize
  }
}

\newcommand{\tool}{\func{AIOCJ}}

\newcommand{\ambient}{A}
\newcommand{\ambientN}{\rules}

\setMyChor

\begin{document}

\title{Dynamic Choreographies\thanks{This work is partly supported by the MIUR FIRB project FACE 
(Formal Avenue for Chasing malwarE) RBFR13AJFT and by the Italian MIUR PRIN Project CINA Prot. 
2010LHT4KM.}}
\subtitle{Safe Runtime Updates of Distributed Applications\iftoggle{tech_report}{\\\vspace{1em}Technical Report}{}}

\author{Mila Dalla Preda\inst{1} \and Maurizio 
Gabbrielli\inst{2} \and Saverio 
Giallorenzo\inst{2} \and \\ Ivan Lanese\inst{2} \and Jacopo Mauro\inst{2} }
\institute{Department of Computer Science - Univ. of Verona
\and Department of Computer Science and Engineering -
 Univ. of Bologna / INRIA
}

\maketitle



\begin{abstract}
Programming distributed applications free from communication deadlocks and
races is complex. Preserving these properties when applications are updated at
runtime is even harder.

We present \AIOC{}, a language for programming distributed
applications that are free from deadlocks and races by construction. A \AIOC{}
program describes a whole distributed application as a unique entity
(choreography). \AIOC{} allows the programmer to specify which parts
of the application can be updated. At runtime, these parts may be
replaced by new \AIOC{} fragments from outside the
application. \AIOC{} programs are compiled, generating code for each
site, in a lower-level language called \APOC{}.  We formalise both
\AIOC{} and \APOC{} semantics as labelled transition systems and
prove the correctness of the compilation as a trace equivalence
result. As corollaries, \APOC{} applications are free from
communication deadlocks and races, even in presence of runtime updates.
 \end{abstract}


\section{Introduction}\label{sec:intro}

Programming distributed applications is an error-prone activity. Participants
send and receive messages and, if the application is badly programmed,
participants may get stuck waiting for messages that never arrive
(communication deadlock), or they may receive messages in an unexpected order,
depending on the speed of the other participants and of the network (races).

Recently, language-based approaches have been proposed to tackle the
complexity of programming concurrent and distributed
applications. Languages such as Rust~\cite{rust} or SCOOP~\cite{scoop}
provide higher-level primitives to program concurrent applications
which avoid by construction some of the risks of concurrent
programming. Indeed, in these settings most of the work needed to
ensure a correct behaviour is done by the language compiler and
runtime support. Using these languages requires a conceptual shift
from traditional ones, but reduces times and costs of development,
testing, and maintenance by avoiding some of the most common
programming errors.

Here, we propose an approach based on \emph{choreographic
  programming}~\cite{hondaESOPext,POPLmontesi,scribble,SEFM08}
following a similar philosophy, tailored for distributed
applications. In choreographic programming, a whole distributed
application is described as a unique entity, by specifying the
expected interactions and their order. For instance, a price request
from a buyer to a seller is written as
\lstinline{priceReq:}~\lstinline{buyer( b_prod )} $\rightarrow$
\lstinline{seller( s_prod )}. It specifies that the \lstinline{buyer}
sends along channel \lstinline{priceReq} the name of the desired
product \lstinline{b_prod} to the \lstinline{seller}, which stores it
in its local variable \lstinline{s_prod}.
Since in choreographic languages sends and receives are always paired,
the coupling of exactly one receive with each send and vice versa
makes communication deadlocks or races impossible to write. 
Given a choreography, a main challenge is to produce low-level
distributed code which correctly implements the desired behaviour.

We take this challenge one step forward: we consider \emph{updatable}
applications, whose code can change while the application is running,
dynamically integrating code from the outside. Such a feature, tricky in a
sequential setting and even more in a distributed one, has countless
uses: deal with emergency requirements, cope with rules and
requirements which depend on contextual properties, improve and
specialize the application to user preferences, and so on. We propose
a general mechanism, which consists in delimiting inside the
application blocks of code, called \emph{scopes}, that may be
dynamically replaced with new code, called \emph{update}.
The details of the behaviour of the updates do not need to be foreseen, updates may even be 
written 
while the application is running.

Runtime code replacement performed using languages not providing
dedicated support is extremely error-prone. For instance, considering
the price request example above, assume that we want to update the
system allowing the buyer to send to the seller also its fidelity card
ID to get access to some special offer. If the buyer is updated first
and it starts the interaction before the seller has been updated, the
seller is not expecting the card ID, which may be sent and lost, or
received later on, when some different message is expected, thus
breaking the correctness of the application.  Vice versa, if the
seller is updated first, (s)he will wait for the card ID, which the
buyer will not send, leading the application to a deadlock.
In our setting, the available updates may change at any time, posing an
additional challenge. Extra precautions are needed to ensure that all
the participants agree on which code is used for a given update.
For instance, in the example above, suppose that the buyer finds the update
that allows the sending of the card ID, and applies this update before the
seller does. If the update is no more available when the seller looks for it,
then the application ends up in an inconsistent state, where the update is
only partially applied, and the seller will receive an unexpected message
containing the card ID.

If both the original application and the updates are programmed using a
choreographic language, these problems cannot arise. In fact, at the
choreographic level, the update is applied atomically to all the involved
participants.  Again, the tricky part is to compile the choreographic code to
low-level distributed code ensuring correct behaviour. In particular, at
low-level, the different participants have to coordinate their updates
avoiding inconsistencies. The present paper proposes a solution to this
problem. In particular:

\begin{itemize}
\item we define a choreographic language, called \AIOC{}, to program
distributed applications and supporting code update (\S~\ref{sec:aioc});
\item we define a low-level language, called \APOC{}, based on standard
send and receive primitives (\S~\ref{sec:apoc});
\item we define a behaviour-preserving projection function compiling \AIOC{}s
into \APOC{}s (\S~\ref{sec:proj});
\item we give a formal proof of the correctness of the projection
  function (\S~\ref{sec:corr}). Correctness is guaranteed even in a scenario
  where the new code used for updates dynamically changes at any
  moment and without notice.
\end{itemize}

The contribution outlined above is essentially theoretical, but it has already
been applied in practice, resulting in \tool, 
an adaptation framework described in~\cite{SLE}.
The theoretical underpinning of \tool{} is a specific instantiation of the results presented here. 
Indeed, \tool{} further specifies how to manage the updates, e.g.,
how to decide when updates should be applied and
which ones to choose if many of them apply.
For more details on the implementation and
more examples we refer the interested reader to the
website~\cite{AIOCJ}.
Note that the user of \tool{} does not need to master all the 
technicalities we discuss here, since they are embedded within \tool{}. In particular,
\APOC{}s and the projection are automatically handled and hidden from the user.

Proofs, additional details, and examples are available in the companion technical report \cite{TR}.

\section{Dynamic Interaction-Oriented Choreography (\AIOC{})}\label{sec:aioc}
This section defines 
the syntax and semantics of the
\AIOC{} language.

 
The languages that we propose rely on a set $\mathit{Roles}$, ranged over by
$r,s,\dots$, whose elements identify the participants in the choreography.  
\iftoggle{tech_report}{
We
call them roles to highlight that they have a specific duty in the
choreography.
Each role owns its local resources.

}{}
Roles exchange messages over channels,
also called \emph{operations}:
\emph{public operations}, ranged over by $o$, and \emph{private
  operations}, ranged over by $o^*$. We use $o^?$ to range over both
public and private operations.
Public operations represent relevant communications inside the
application. We ensure that both the \AIOC{} and the
corresponding \APOC{} perform the same public operations, in the same
order. Vice versa, private communications are used when moving from the \AIOC{} level to the 
\APOC{} level, for synchronisation purposes.  We denote
with $\mathit{Expr}$ the set of expressions, ranged over by $e$. We
deliberately do not give a formal definition of expressions and of their
typing, since our results do not depend on it. We only require that
expressions include at least values, belonging to a set $\mathit{Val}$
ranged over by $v$, and variables, belonging to a set
$\mathit{Var}$ ranged over by $x,y,\dots$. 
We also assume a set of boolean expressions 
ranged over by $b$.

The syntax of \emph{\AIOC{} processes}, ranged over by ${\mathcal I}, {\mathcal I}', \ldots$, is 
defined as follows:
\[
\begin{array}{ll}
{\mathcal I} \gram &\comm{o^?}{r_1}{e}{r_2}{x} 
\mid {\mathcal I}\seqOp{\mathcal I}'\mid {\mathcal I}\parOpI{\mathcal I}'\mid \assign{x}{r}{e}
\mid \one \mid \zero \mid\\ 
& \ifthen{b \at r}{\mathcal{I}}{\mathcal{I}'} \mid \while{b \at r}{\mathcal{I}} \mid 
\scope{l}{r}{\mathcal{I}}{}
\end{array}
\]
Interaction $\comm{o^?}{r_1}{e}{r_2}{x}$ means that role $r_1$ sends a
message on operation $o^?$ to role $r_2$ (we require $r_1 \neq
r_2$). The sent value is obtained by evaluating expression $e$ in the
local state of $r_1$ and it is then stored in variable $x$ in $r_2$. 
Processes ${\mathcal
  I}\seqOp{\mathcal I}'$ and ${\mathcal I}\parOpI{\mathcal I}'$ denote sequential and parallel composition. 
Assignment $\assign{x}{r}{e}$ assigns the evaluation of
expression $e$ in the local state of $r$ to its local variable $x$.  
The empty process $\one$ defines a \AIOC{} that can only terminate. 
$\zero$ represents a terminated \AIOC{}. It is needed for the
definition of the operational semantics and it is not intended to be used by the programmer.
We call
\emph{initial} a \AIOC{} process where $\zero$ never occurs.
Conditional $\ifthen{b \at r}{\mathcal{I}}{\mathcal{I}'}$ and iteration  $\while{b \at r}{\mathcal{I}}$ are guarded by the evaluation of boolean expression $b$ in the local state of $r$. 
 The construct $\scope{l}{r}{\mathcal{I}}{\Delta}$ delimits a subterm $\mathcal{I}$ of the \AIOC{} 
process that may be updated in the future.
%
%
In
$\scope{l}{r}{\mathcal{I}}{\Delta}$, role $r$ coordinates the
updating procedure by interacting with
the other roles involved in
the scope.


\AIOC{} processes do not execute in isolation: they are equipped with a
\emph{global state} $\Sigma$ and a set of (available) updates $\rules$.
%
A global state $\Sigma$ is a map that defines the value $v$ of each variable $x$ in a given role $r$, namely $\Sigma: \mathit{Roles} \times \mathit{Var} \rightarrow
\mathit{Val}$. 
The local state of role $r$ is $\Sigma_r:\mathit{Var} \rightarrow \mathit{Val}$ and it
verifies $\forall x \in \mathit{Var}: \;
\Sigma(r,x) = \Sigma_r(x)$. Expressions are always evaluated by a
given role $r$: we denote the evaluation of expression $e$ in local state $\Sigma_r$ as
$\grass{e}_{\Sigma_r}$. We assume $\grass{e}_{\Sigma_r}$ is always
defined (e.g., an error value is
given as a result if evaluation is not possible) and that for each boolean expression $b$,
$\grass{b}_{\Sigma_r}$ is either $\tt{true}$ or $\tt{false}$.
$\rules$ denotes a set of updates, i.e., \AIOC{}s that may replace a scope.
$\rules$ may change at runtime.

Listing~\ref{es:bank} gives a realistic example of \AIOC{} process where a \buyer\, orders a
product from a \seller, paying via a \bank.
%
%
%
\iftoggle{tech_report}{\begin{figure}[tb]\begin{center}}{}
\begin{lstlisting}[mathescape=true,numbers=left,numbersep=1em,xleftmargin=2em,caption=\AIOC{} process for Buying Scenario.,label=es:bank]
price_ok@buyer = false; continue@buyer = true;
while ( !price_ok and continue )@buyer {
 b_prod@buyer = getInput();
 priceReq : buyer( b_prod ) $\rightarrow$ seller( s_prod );
 scope @seller {
   s_price@seller = getPrice( s_prod );
   offer : seller( s_price ) $\rightarrow$ buyer( b_price )
 };
 price_ok@buyer = getInput();
 if ( !price_ok )@buyer {
  continue@buyer = getInput()} };
if ( price_ok )@buyer {
 payReq : seller( payDesc( s_price ) ) $\rightarrow$ bank( desc );
 scope @bank {
   payment_ok@bank = true;
   pay : buyer( payAuth( b_price ) ) $\rightarrow$ bank( auth );
   ... // code for the payment
 };
 if ( payment_ok )@bank {
   confirm : bank( null ) $\rightarrow$ seller( _ ) |
   confirm : bank( null ) $\rightarrow$ buyer( _ )
 } else { abort : bank( null ) $\rightarrow$ buyer( _ ) } }
\end{lstlisting}
\iftoggle{tech_report}{\end{center}\end{figure}}{\vspace{-3em}}
%
%
Before starting the application by iteratively asking the price of
some goods to the \seller{}, the \buyer{} at Line 1 initializes
its local variables \lstinline{price_ok} and
\lstinline{continue}. Then, by using function \lstinline{getInput} (Line
3) (s)he reads from the local console the name of the product to buy and, at
Line 4, engages in a communication via operation \lstinline{priceReq}
with the \seller.
The \seller\, computes the
price of the product calling the function \lstinline{getPrice} (Line 6) and, via operation 
\lstinline{offer}, it sends the price to the 
\buyer{} (Line 7), that stores it in a
local variable \lstinline{b_price}. These last two operations are performed within a scope, 
allowing this code to be updated in the future to deal with changing business rules.
If the offer is accepted, the \seller\, sends to the \bank\, the
payment details (Line 13).  The \buyer\, then authorises the payment
via operation \pay. We omit the details of the local execution of the
payment at the \bank. Since the payment may be critical for security
reasons, the related communication is enclosed in a scope (Lines
14-18), thus allowing the introduction of a more refined procedure
later on. After the scope successfully terminates, the application
ends with the \bank\, acknowledging the payment to the \seller\, and
the \buyer\, in parallel (Lines 20-21). If the payment is not
successful, the failure is notified to the \buyer{} only.
Note that at Line 1, the annotation \lstinline{@buyer} means that the variables belong to the 
\buyer{}.  Similarly, at Line 2, the annotation \lstinline{@buyer} means that the guard of the 
while is 
evaluated by \lstinline{buyer}.
The term \lstinline{@seller} in Line 5 instead, being part of the scope construct, 
indicates the participant that coordinates the code update.

Assume now that the seller
direction decides to 
define new business
rules. For instance, the seller may distribute a fidelity card to
buyers, allowing them to get a 10\% discount on their
purchases. This business need can be
faced by adding the \AIOC{} below to the set of available updates, so that it can be used to 
replace the scope at Lines 5-8 in Listing \ref{es:bank}.

\iftoggle{tech_report}{\begin{figure}[tb]\begin{center}}{}
\begin{lstlisting}[mathescape=true,numbers=left,numbersep=1em,xleftmargin=2em,caption=Fidelity Card 
Update, label=rule_price_inquiry]
cardReq : seller( null ) $\rightarrow$ buyer( _ );
card_id@buyer = getInput();
cardRes : buyer( card_id ) $\rightarrow$ seller( buyer_id );
if isValid( buyer_id )@seller {
 s_price@seller = getPrice( s_prod ) * 0.9 
} else { s_price@seller = getPrice( s_prod ) };
offer : seller( s_price ) $\rightarrow$ buyer( b_price )
\end{lstlisting}
\iftoggle{tech_report}{\end{center}\end{figure}}{\vspace{-3em}}
When this code executes, the \seller\,
asks the card ID to the \buyer.  The \buyer\, inputs the ID, stores it
into the variable \texttt{card\_id} and sends this
information to the \seller.  If the card ID is valid then the
discount is applied, otherwise the standard price is computed.
%

\subsection{Connectedness} 
\label{sub:connectedness}

In order to prove our main result, we require the \AIOC{} code of the updates
and of the starting programs to satisfy a well-formedness syntactic condition
called \emph{connectedness}. This condition is composed by \emph{connectedness
for sequence} and \emph{connectedness for parallel}. 
Intuitively, connectedness for sequence ensures that the \APOC{}
network obtained by projecting a sequence ${\mathcal I}\seqOp{\mathcal
  I}'$ executes first the actions in ${\mathcal I}$ and then those in
${\mathcal I}'$, thus respecting the intended semantics of sequential
composition.
Connectedness for parallel prevents interferences between parallel 
interactions.
To formally define connectedness we introduce, in Table~\ref{table:tri},
the auxiliary functions $\transI$ and $\transF$
that, given a \AIOC{} process, compute sets of pairs representing senders
and receivers of possible initial and final interactions in its execution.
We represent one such
pair as $\pair{r_1}{r_2}$. Actions located at $r$ are represented as $\pair{r}{r}$. For instance,
given an interaction $\comm{o^?}{r_1}{e}{r_2}{x}$ both its $\transI$ and 
$\transF$ are $\{\pair{r_1}{r_2}\}$. For conditional, $\transI( 
\ifthen{b \at r}{\mathcal{I}}{\mathcal{I}'}) =  \{\pair{r}{r}\}$ since the first action 
executed is the evaluation of the guard by role $r$.
The set $\transF(\ifthen{b \at r}{\mathcal{I}}{\mathcal{I}'})$ is normally $\transF({\mathcal I}) 
\cup 
\transF({\mathcal I}')$, since the execution terminates with an action from one of the 
branches. If instead the branches are both empty then $\transF$ is 
$\{ \pair{r}{r} \}$, representing guard evaluation. 
\begin{table*}[t]
\centering
{\footnotesize
\begin{tabular}{ll}
$\transI(\comm{o^?}{r_1}{e}{r_2}{x}) = \transF(\comm{o^?}{r_1}{e}{r_2}{x}) = \{\pair{r_1}{r_2}\} $ 
\\
$\transI(\assign{x}{r}{e}) = \transF(\assign{x}{r}{e}) = \{\pair{r}{r}\}$ \\
$\transI(\one)=\transI(\zero)=\transF(\one)=\transF(\zero)=\emptyset$\\
$\transI({\mathcal I} \parOpI {\mathcal I}') = \transI({\mathcal I}) \cup \transI({\mathcal 
I}') \hfill \transF({\mathcal I} \parOpI {\mathcal I}') =  \transF({\mathcal I}) \cup 
\transF({\mathcal 
I}')$\\
$\transI({\mathcal I}\seqOp{\mathcal I}')  = 
\left \{ 
\begin{array}{ll}
\transI({\mathcal I}') & \mbox{ if } \transI({\mathcal I}) = \emptyset\\
\transI({\mathcal I}) & \mbox{ otherwise}\\
\end{array}
\right.
$ \hfill \quad
$\transF({\mathcal I}\seqOp{\mathcal I}')  = 
\left \{ 
\begin{array}{ll}
\transF({\mathcal I}) & \mbox{ if } \transF({\mathcal I}') = \emptyset\\
\transF({\mathcal I}') & \mbox{ otherwise}\\
\end{array}
\right.
$\\
$\transI( \ifthen{b \at r}{\mathcal{I}}{\mathcal{I}'}) =  \transI(\while{b \at r}{\mathcal{I}} ) =  
\{\pair{r}{r}\}$\\
$\transF( \ifthen{b \at r}{\mathcal{I}}{\mathcal{I}'}) = 
\left \{ 
\begin{array}{ll}
\{ \pair{r}{r} \} & \mbox{ if } \transF({\mathcal I}) \cup \transF({\mathcal I}') = \emptyset\\
\transF({\mathcal I}) \cup \transF({\mathcal I}') & \mbox{ otherwise}\\
\end{array}
\right.
$\\
$\transF(\while{b \at r}{\mathcal{I}} ) = \left \{ 
\begin{array}{ll}
\{ \pair{r}{r} \} & \mbox{ if } \transF({\mathcal I}) = \emptyset\\
\transF({\mathcal I}) & \mbox{ otherwise}\\
\end{array}
\right.$\\
$\transI(\scope{l}{r}{\mathcal I}{\Delta}) = \{ \pair{r}{r}\}$\\
$\transF(\scope{l}{r}{\mathcal I}{\Delta}) = 
\left \{ 
\begin{array}{ll}
\{ \pair{r}{r} \} & \mbox{ if } \roles({\mathcal I}) \subseteq \{ r \}\\
\bigcup_{r' \in 
\roles(\mathcal{I}) \smallsetminus \{r\}} \{ \pair{r'}{r} \} & \mbox{ otherwise}\\
\end{array}
\right.
$
\end{tabular}
}
\vspace{5pt}
\caption{Auxiliary functions $\transI$ and $\transF$.}\label{table:tri} 
\end{table*}

\iftoggle{tech_report}{
We assume 
a function $\roles({\mathcal I})$ that computes the 
roles of a \AIOC{} process ${\mathcal I}$ defined as follows:

\vspace{0.2cm}
\begin{tabular}{l}
$\roles(\comm{o^?}{r_1}{e}{r_2}{x}) =  \{r_1,r_2\}$ \\
$\roles(\one) = \roles(\zero) = \emptyset$ \\
$\roles(x \at r = e) = \{r\}$\\
$\roles({\mathcal I}\seqOp{\mathcal I}') = \roles({\mathcal I} | {\mathcal I}') = \roles({\mathcal I}) 
\cup \roles({\mathcal I}')$\\
$\roles(\ifthen{b \at r}{\mathcal{I}}{\mathcal{I}'}) = \{r\} \cup \roles({\mathcal{I}}) \cup 
\roles(\mathcal{I}')$\\
$\roles(\while{b \at r}{\mathcal{I}}) = \{r\} \cup 
\roles(\mathcal{I})$ \\
$\roles(\scope{l}{r}{\mathcal I}{\Delta}) = \{r\} \cup 
\roles(\mathcal{I})$
\end{tabular}
\vspace{0.2cm}

}{ 
We assume 
a function $\roles({\mathcal I})$ that computes the 
roles of a \AIOC{} process ${\mathcal I}$.
}
We also assume a function $\operations$ that given a \AIOC{} process returns the set of signatures 
of 
its interactions, where the signature of interaction $\comm{o^?}{r_1}{e}{r_2}{x}$ is 
$\sign{o^?}{r_1}{r_2}$.
\iftoggle{tech_report}{
It can be inductively defined as follows:

\vspace{0.2cm}
\begin{tabular}{l}
$\operations(\comm{o^?}{r_1}{e}{r_2}{x}) = \{ \sign{o^?}{r_1}{r_2} \}$\\
$\operations({\mathcal I} \parOpI {\mathcal I}') = \operations({\mathcal I} \seqOp  {\mathcal I}') = \operations({\mathcal I}) \cup 
\operations({\mathcal I}')$\\
$\operations( \ifthen{b \at r}{\mathcal{I}}{\mathcal{I}'}) 
= \operations({\mathcal I}) \cup 
\operations({\mathcal I}')$\\
$\operations( \scope{l}{r}{\mathcal I}{\Delta})
= \operations({\mathcal I})$\\ 
$\operations( \while{b \at r}{\mathcal{I}}) 
= \operations({\mathcal I})$\\
$\operations(\assign{x}{r}{e}) = \operations(\one) = \operations(\zero) = \emptyset$\\
\end{tabular}
\vspace{0.2cm}
}{ 
For a formal definition of the functions $\roles$ and $\operations$ we refer the reader 
to the companion technical report~\cite{TR}.
}

\begin{definition}[Connectedness]\label{def:connectedness}
A \AIOC{} process ${\mathcal I}$ is connected if it satisfies:
\begin{itemize}
\item {\bf connectedness for sequence:} each
subterm of the form ${\mathcal I}' \seqOp {\mathcal I}''$ satisfies
$\forall \pair{r_1}{r_2} \in \transF({\mathcal I}'), \forall
\pair{s_1}{s_2} \in \transI({\mathcal I}'') \; . \;
\{r_1,r_2\} \cap \{s_1,s_2\} \neq \emptyset$;
\item {\bf connectedness for parallel:}  each
subterm of the form~ ${\mathcal I}' \parOpI {\mathcal I}''$ satisfies
$\operations({\mathcal I}') \cap \operations({\mathcal I}'') = \emptyset$.
\end{itemize}
\end{definition}

%
Requiring connectedness does not hamper programmability, since it naturally
holds in most of the cases (see, e.g., \cite{SLE,AIOCJ}), and it can always be
enforced automatically restructuring the \AIOC{} while preserving its
behaviour, following the lines of~\cite{WWVlanese}.
Also, connectedness can be checked efficiently.
\begin{theorem}[Connectedness-check complexity]\label{teo:compl}\mbox{}\\
The connectedness of a \AIOC{} process ${\mathcal I}$ can be checked in
time $O(n^2\log(n))$, where $n$ is the number of nodes in the abstract syntax
tree of ${\mathcal I}$.
\end{theorem}
\iftoggle{tech_report}{
The proof of the theorem is reported in Appendix \ref{app:complexity}.}{}

Note that we allow only connected updates. Indeed, replacing a scope
with a connected update always results in a deadlock- and race-free
\AIOC{}. Thus, there is no need to perform expensive runtime checks to
ensure connectedness of the application after an arbitrary sequence of
updates has been applied.

\subsection{\AIOC{} semantics} 
\label{sub:aioc_semantics}

We can now define \AIOC{} systems and their semantics.

\begin{definition}[\AIOC{} systems]
A \AIOC{} system is a triple $\tuple{\Sigma, \rules,\mathcal{I}}$
denoting a \AIOC{} process $\mathcal{I}$ equipped with a global state $\Sigma$ and a set of 
updates $\rules$.
\end{definition}
%
%
\begin{table*}[t]
{\footnotesize
\[
\begin{array}{c}
\mathrule{Interaction}{
\grass{e}_{\Sigma_{r_1}} = v
}{\tuple{\ambient, \comm{o^?}{r_1}{e}{r_2}{x}} 
\arro{\commLabel{o^?}{r_1}{v}{r_2}{x}} \tuple{\ambient,\assign{x}{r_2}{v}}}
\hfill\qquad
\mathrule{Sequence}{\tuple{\ambient,{\mathcal I}} \arro{\mu} \tuple{\ambient',{\mathcal I}'} \hfill 
\mu \neq \tick}{\tuple{\ambient,{\mathcal I}\seqOp{\mathcal J}} \arro{\mu} \tuple{\ambient',{\mathcal 
I}'\seqOp{\mathcal J}}}
\\
\mathrule{Assign}{
\grass{e}_{\Sigma_r} = v
}{
\tuple{\state,\rules,{\assign{x}{r}{e}}}
\arro{\tau}
\tuple{\state\substLabel{v}{x}{r},\rules, \one}
}
\hfill
\mathrule{Seq-end}{\tuple{\ambient,{\mathcal I}} \arro{\tick} \tuple{\ambient,{\mathcal I}'} \hfill 
\tuple{\ambient,{\mathcal J}} \arro{\mu} \tuple{\ambient,{\mathcal J}'}}{\tuple{\ambient,{\mathcal 
I}\seqOp{\mathcal J}} \arro{\mu} \tuple{\ambient,{\mathcal J}'}}
\\
\iftoggle{tech_report}{
\mathrule{Parallel}{\tuple{\ambient,{\mathcal I}} \arro{\mu} \tuple{\ambient',{\mathcal I}'} \hfill 
\mu \neq \tick}{\tuple{\ambient,{\mathcal I}\parallel{\mathcal J}} \arro{\mu} 
\tuple{\ambient',{\mathcal I}'\parallel{\mathcal J}}}
\hfill
\mathrule{Par-end}{\tuple{\ambient,{\mathcal I}} \arro{\tick} \tuple{\ambient,{\mathcal I}'} \hfill 
\tuple{\ambient,{\mathcal J}} \arro{\tick} \tuple{\ambient,{\mathcal 
J}'}}{\tuple{\ambient,{\mathcal I}\parallel{\mathcal J}} \arro{\tick} \tuple{\ambient,{\mathcal 
I}'\parallel{\mathcal J}'}}
\\
\mathrule{If-then}{
\grass{b}_{\Sigma_r} = \tt{true}}
{\tuple{\ambient,\ifthen{b \at r}{\mathcal{I}}{\mathcal{I}'}} 
\arro{\tau} \tuple{\ambient,{\mathcal I}}} \hfill
\mathrule{If-else}{
\grass{b}_{\Sigma_r} = 
\tt{false}}{\tuple{\ambient,\ifthen{b \at r}{\mathcal{I}}{\mathcal{I}'}} 
\arro{\tau} \tuple{\ambient,{\mathcal I}'}}
\\
\mathrule{While-unfold}{
\grass{b}_{\Sigma_r} = 
\tt{true}}{\tuple{\ambient,\while{b \at r}{\mathcal{I}}} \arro{\tau} 
\tuple{\ambient,{\mathcal I} \seqOp \while{b \at r}{\mathcal{I}}}} \hfill
\mathrule{While-exit}{
\grass{b}_{\Sigma_r} = 
\tt{false}}{\tuple{\ambient,\while{b \at r}{\mathcal{I}}} \arro{\tau} \tuple{\ambient,\one}} 
\\
}{
}
\mathrule{\AdaptRule}{ 
\roles({\mathcal{I}'}) \subseteq \roles({\mathcal{I})} \quad \mathcal{I}'
\in \rules
\quad \mathcal{I}' \textrm{ connected}
}{\tuple{\ambient,\scope{l}{r}{\mathcal 
I}{\Delta}} \arro{{\mathcal I}'} \tuple{\ambient, {\mathcal 
I}'}}
\hfill
\mathax{\NoAdaptRule}{\tuple{\ambient,\scope{l}{r}{\mathcal I}{\Delta}}
\arro{\texttt{\NoAdaptLabel}} \tuple{\ambient,{\mathcal I}}}
\\
\mathax{End}{\tuple{\ambient,\one} \arro{\tick} \tuple{\ambient,\zero}}
\hfill
\mathax{Change-Updates}{\tuple{\state,\rules,{\mathcal I}} \arro{\rules'} 
\tuple{\state,\rules',{\mathcal I}}}
\end{array}
\]
}
\caption{\iftoggle{tech_report}{\AIOC{} system semantics.}{\AIOC{} system 
semantics (excerpt).}}\label{table:ioclts} 
\end{table*}

\iftoggle{tech_report}{
\begin{definition}[\AIOC{} systems semantics]
The semantics of \AIOC{} systems is defined as the smallest labelled
transition system (LTS) closed under the rules in
Table~\ref{table:ioclts}, where symmetric rules for parallel
composition have been omitted.
\end{definition}
} 
{ 
\begin{definition}[\AIOC{} systems semantics]
The semantics of \AIOC{} systems is defined as the smallest labelled
transition system (LTS) closed under the rules in
Table 6 in the companion technical report \cite{TR} (excerpt in Table~\ref{table:ioclts}), 
where symmetric rules for parallel composition have been omitted.
\end{definition}
} 
The rules in Table \ref{table:ioclts} describe the behaviour of a \AIOC{} system 
by induction on the structure of its \AIOC{} process.  We use $\mu$ to range over
labels. Also, we use $\ambient$ as an abbreviation for
$\state,\rules$. 
\iftoggle{tech_report}{
Rule \ruleName{Interaction}  executes a communication from $r_1$ to $r_2$ on
operation $o^?$,  where $r_1$ sends to $r_2$ the value $v$ of an expression
$e$. The value $v$ is then stored in $x$ by $r_2$.
Rule \ruleName{Assign} evaluates the expression $e$ in the local state
$\Sigma_r$ and stores the resulting value $v$ in the local variable $x$ in
role $r$ ($\substLabel{v}{x}{r}$ represents the substitution).
Rule \ruleName{End}  terminates the execution of an empty process.  Rule
\ruleName{Sequence} executes a step in the first process of a sequential
composition, while rule \ruleName{Seq-end} acknowledges the termination of the
first process, starting the second one.
Rule \ruleName{Parallel} allows a process in a parallel composition to
compute, while rule \ruleName{Par-end} synchronises the termination of two
parallel processes.
Rules \ruleName{If-then} and \ruleName{If-else} evaluate the boolean guard of
a conditional, selecting the then and the else branch, respectively.
Rules \ruleName{While-unfold} and \ruleName{While-exit} correspond
respectively to the unfolding of a while when its condition is satisfied and
to its termination otherwise. }
{ We comment below on the main rules.

Rule \ruleName{Interaction} executes a communication from $r_1$ to $r_2$ on
operation $o^?$,  where $r_1$ sends to $r_2$ the value $v$ of an expression
$e$. The value $v$ is then stored in $x$ by $r_2$.
Rule \ruleName{Assign} evaluates the expression $e$ in the local state
$\Sigma_r$ and stores the resulting value $v$ in the local variable $x$ in
role $r$ ($\substLabel{v}{x}{r}$ represents the substitution). }
The rules \ruleName{\AdaptRule} and \ruleName{\NoAdaptRule} deal with the code
replacement and thus the application of an update. Rule \ruleName{\AdaptRule} models
the application of the update $\arule{l}{\cond}{\mathcal{I}'}$ to the scope
$\scope{l}{r}{\mathcal I}{\Delta}$ which, as a result, is replaced by the \AIOC{}
process ${\mathcal{I}'}$. This rule requires the update to be connected.
Rule \ruleName{\NoAdaptRule} removes the scope
boundaries and starts the execution of the body of the scope.
Rule \ruleName{Change-Updates} allows the set $\rules$ of available updates to 
change. This rule is
always enabled since its execution can happen at any time and the
application cannot forbid it.

In our theory, whether to update a scope or not, and which update to apply if many are 
available, is completely non-deterministic. We have adopted
this view to maximize generality. However, for practical
applications, one needs rules and conditions which define when an update has to
be performed. Refining the semantics to introduce rules for decreasing (or eliminating) 
the non-determinism would not affect the
correctness of our approach.
One such refinement has been explored in~\cite{SLE}.

We define \emph{\AIOC{} traces}, where all the performed actions are observed,
and \emph{weak \AIOC{} traces}, where interactions on private operations and
silent actions $\tau$ are not visible.
%
\begin{definition}[\AIOC{} traces]
A (strong) trace of a \AIOC{} system $\tuple{ \Sigma_1, \rules_1,{\mathcal I}_1}$ is a sequence 
(finite or infinite) of
  labels $\mu_1, \mu_2, \dots$ such that there is a sequence of \AIOC{}
  system transitions $\tuple{\Sigma_1,\rules_1,{\mathcal I}_1}
  \arro{\mu_1} \tuple{\Sigma_2,\rules_2,{\mathcal I}_2} \arro{\mu_2} \dots$.\\
A weak trace of a \AIOC{} system $\tuple{
  \Sigma_1,\rules_1,{\mathcal I}_1}$ is a sequence of labels
  $\mu_1, \mu_2, \dots$ obtained by removing all the labels
  corresponding to private communications, i.e.,~of the form
  $\commLabel{o^*}{r_1}{v}{r_2}{x}$, and the silent labels $\tau$ from a
  trace of $\tuple{ \Sigma_1, \rules_1,{\mathcal I}_1}$.
\end{definition}

\section{Dynamic Process-Oriented Choreography (\APOC{})}\label{sec:apoc}

This section describes the syntax and operational semantics of
\APOC{}s. \APOC{}s include \emph{processes}, ranged over by $P$, $P'$,
$\ldots$, describing the behaviour of participants.  $(P,\Gamma)_r$
denotes a \emph{\APOC{} role} named $r$, executing process $P$ in a
local state $\Gamma$. \emph{Networks}, ranged over by $\net$, $\net'$,
$\ldots$, are parallel compositions of \APOC{} roles with different
names. \APOC{} systems, ranged over by ${\mathcal S}$, are \APOC{} networks
equipped with a set of updates $\rules$, namely pairs $\tuple{\rules, \net}$.
\[
\begin{array}{ll}
P \gram & \cinp{o^?}{x}{r} \mid \cout{o^?}{e}{r} \mid \cout{o^\ast}{X}{r} \mid
P \seqOp P'\mid \ P \parOpP P' \ \mid  x = e \mid \while{b}{P}\\ 
& \hspace{-1.5em} \mid \ifthen{b}{P}{P'} \mid 
\pscope{n}{l}{r}{P}{\Delta}{S} \mid \psscope{n}{l}{r}{P} \mid \one \mid \zero\\
\end{array}
\]
$
X \gram  {\tt{no}} \mid P \hfill
\net \gram  \role{P,\Gamma}{r}\mid \net \parallel \net' \hfill
{\mathcal S}  \gram  \tuple{\rules, \net}
$\\\\
\noindent
Processes include receive action $\cinp{o^?}{x}{r}$ on a specific
operation $o^?$ (either public or private) of a message from role $r$
to be stored in variable $x$, send action $\cout{o^?}{e}{r}$ of an
expression $e$ to be sent to role $r$, and higher-order send action
$\cout{o^*}{X}{r}$ of the higher-order argument $X$ to be sent to role
$r$. Here $X$ may be either a \APOC{} process $P$, which is the new code
for a scope in $r$, or a token $\tt no$, notifying that no update
is needed. $P \seqOp P'$ and $P \parOpP P'$ denote the sequential and parallel
composition of $P$ and $P'$, respectively. Processes also feature
assignment $x = e$ of expression $e$ to variable $x$, the process
$\one$, that can only successfully terminate, and the terminated
process $\zero$. We also have conditionals $\ifthen{b}{P}{P'}$
and loops $\while{b}{P}$. 
Finally, we have two constructs for scopes. Scope
$\pscope{n}{l}{r}{P}{\Delta}{S}$ may occur only inside role $r$ and
acts as coordinator to apply (or not apply) the update. The
shorter version $\psscope{n}{l}{r}{P}$ is used instead when the role is not
the coordinator of the scope.
In fact, only the coordinator needs to
know the set
$S$ of involved roles to communicate which update to apply.
Note that scopes are prefixed by an index $n$. Indexes are unique in each role and are used
to avoid interference between different scopes in the same role.

\subsection{Projection}\label{sec:proj}
Before defining the semantics of \APOC{}s, we define the
projection of a \AIOC{} process onto \APOC{} processes. This is
needed to define the semantics of updates at the \APOC{} level. The projection
exploits auxiliary communications to coordinate the different
roles, e.g., ensuring that in a conditional they all select the
same branch. To define these auxiliary communications and avoid
interference, it is convenient to annotate \AIOC{} main constructs with
unique indexes.

\begin{table*}[t]
\footnotesize
$\begin{array}{ll}
\begin{array}{ll}
	\mbox{\rulebox{$\pi(\one,s)$} = $\one$} \; \qquad \mbox{\rulebox{$\pi(\zero,s)$} = $\zero$} 
	\\
	\mbox{\rulebox{$\pi({\mathcal I} \seqOp {\mathcal I}',s)$} = $\pi({\mathcal I},s) \seqOp \pi({\mathcal I}',s)$}
	\\
	\mbox{\rulebox{$\pi({\mathcal I} \parOpI {\mathcal I}',s)$} = $\pi({\mathcal I},s) \mid \pi({\mathcal I}',s)$}
	\end{array}
\left.
	\mbox{\rulebox{$\pi(n: \assign{x}{r}{e},s)$}} =
	\left \{
	\begin{array}{ll}
	x = e & \mbox{ if } s=r\\
	\one & \mbox{ otherwise}
	\end{array}\right.
\right.
\\
\mbox{\rulebox{$\pi(n:\comm{o^?}{r_1}{e}{r_2}{x} ,s)$}}  =
\left \{ 
\begin{array}{ll}
	\cout{o^?}{e}{r_2} & \mbox{ if } s=r_1\\
	\cinp{o^?}{x}{r_1} & \mbox{ if } s= r_2\\
	\hfill \one \hfill & \mbox{ otherwise}\\
\end{array}
\right.
\end{array}$
\\
$\mbox{\rulebox{$\pi(\ifthenKey{b \at r}{\mathcal{I}}{\mathcal{I}'}{n},s)$}} = \\ \makebox[1em]{}
\left \{ 
\begin{array}{ll}
	\begin{array}{ll}
	\code{if} \; b \; 
		\{ (\prod_{r' \in \roles(\mathcal{I} , \mathcal{I}')\smallsetminus\{r\}} \; 
			\cout{o^*_n}{\mathit{true}}{r'}) \seqOp \; \pi(\mathcal{I},s) \}  \\
			\vspace{0.5em}
		\quad \code{else} \;  \{(\prod_{r' \in \roles(\mathcal{I} 
, \mathcal{I}')\smallsetminus \{r\}} \; 
		\cout{o^*_n}{\mathit{false}}{r'}) \seqOp \; \pi(\mathcal{I}',s)\}
	\end{array}
	& \mbox{if } \mathit{s}=\mathit{r}
\\
	\begin{array}{ll}
	\cinp{o^*_n}{x_n}{r} \seqOp \; \ifthen{x_n}{\pi(\mathcal{I},s)}{\pi(\mathcal{I}',s)}
	\vspace{0.5em}
	\end{array}
	& \mbox{if } r \in \roles(\mathcal{I} , \mathcal{I}') \smallsetminus\{s\}
\\
\hfill \one \hfill & \mbox{otherwise}\\ 
\end{array}
\right.$
\\
$\mbox{\rulebox{$\pi(\whileKey{b \at r}{\mathcal{I}}{n},s)$}} = \\ \makebox[1em]{}
\left \{
\begin{array}{ll}
	\begin{array}{ll}
	\code{while} \; b \; \{(\prod_{r' \in \roles(\mathcal{I}) \smallsetminus \{r\}} 
	\cout{o^*_n}{\mathit{true}}{r'}) \seqOp \pi(\mathcal{I},s) \seqOp \; \\
	\hspace{4em}\prod_{r' \in \roles(\mathcal{I}) \smallsetminus \{r\}} \; \cinp{o^*_n}{\_}{r'}\}
	\seqOp \\
	\hspace{3.5em}\prod_{r'\in \roles(\mathcal{I}) \smallsetminus \{r\}} \; \cout{o^*_n}{\mathit{false}}{r'}
	\vspace{0.5em}
	\end{array}
	& \mbox{if } \mathit{s}=\mathit{r}
\\
	\begin{array}{ll}
	\cinp{o^*_n}{x_n}{r} \seqOp \\ 
	\hspace{1em} \while{x_n}{\pi(\mathcal{I},s) \seqOp \;
	\cout{o^*_n}{\texttt{ok}}{r} \seqOp \cinp{o^*_n}{x_n}{r}}
	\vspace{0.5em}
	\end{array}
	& \mbox{if } s \in\roles(\mathcal{I}) \smallsetminus \{r\}
\\
\hfill \one \hfill & \mbox{otherwise}
\end{array}
\right.$
\\[1mm]
$\mbox{\rulebox{$\pi(n: \scope{l}{r}{\mathcal I}{\Delta},s)$}} =
\left \{ 
\begin{array}{ll}
	\begin{array}{ll}
	\pscope{n}{l}{r}{\pi({\mathcal I},s)}{\Delta}{\roles({\mathcal I})}
	\vspace{0.5em}
	\end{array}
	& \mbox{if } \mathit{s}=\mathit{r}
\\
	\begin{array}{ll}
	\psscope{n}{l}{r}{\pi({\mathcal I},s)}
	\vspace{0.5em}
	\end{array}
	& \mbox{if } s \in \roles(\mathcal{I}) \smallsetminus \{r\} 
\\
\hfill \one \hfill & \mbox{otherwise}\\ 
\end{array}
\right.$
\caption{Process-projection function $\pi$.}\label{table:pi}
\end{table*}

\begin{definition}[Well-annotated \AIOC{}]\label{def:annAIOC}
Annotated \AIOC{} processes are obtained by indexing every interaction,
assignment, scope, and if and while constructs in a \AIOC{} process with a 
natural number $n \in \mathbb{N}$, resulting in the following grammar:
\[
\begin{array}{ll}
{\mathcal I} \gram & n: \comm{o^?}{r_1}{e}{r_2}{x} 
\mid {\mathcal I}\seqOp{\mathcal I}'
\mid {\mathcal I}\parOpI{\mathcal I}'
\mid \one \mid \zero %
\mid n: \assign{x}{r}{e} \\
& \hspace{-.8em}\mid n: \while{b \at r}{\mathcal{I}}
\mid n: \ifthen{b \at r}{\mathcal{I}}{\mathcal{I}'} 
\mid n:\scope{l}{r}{\mathcal{I}}{\Delta}{A}
\end{array}
\]
A \AIOC{} process is well-annotated if all its indexes are
distinct.
\end{definition}
Note that we can always annotate a \AIOC{} process to make it well-annotated.

We now define the \emph{process-projection function} that derives \APOC{}
processes from \AIOC{} processes. 
Given an annotated \AIOC{}
process $\mathcal{I}$ and a role $s$, the projected \APOC{} process
$\pi(\mathcal{I},s)$ is defined by structural induction on ${\mathcal I}$ 
in Table~\ref{table:pi}. Here, with a little abuse of notation, we write
$\roles(\mathcal{I},\mathcal{I}')$ for $\roles(\mathcal{I})\cup\roles(\mathcal{I}')$.
We assume that operations $o^*_n$ and variables $x_n$ are never used in the projected 
\AIOC{} and we use them for auxiliary synchronisations.
 In most of the cases the projection is
trivial. For instance, the projection of an interaction is an output
on the sender role, an input on the receiver, and $\one$ on any other
role. For a conditional
$\ifthenKey{b \at r}{\mathcal{I}}{\mathcal{I}'}{n}$, role $r$ locally
evaluates the guard and then sends its value to the other roles
using auxiliary communications. Similarly, in a loop
$\whileKey{b \at r}{\mathcal{I}}{n}$ role $r$ communicates the evaluation
of the guard to the other roles. Also, after an iteration has
terminated, role $r$ waits for the other roles to terminate and then starts
a new iteration. In both the conditional and the loop, indexes are used
to choose names for auxiliary operations: the choice is coherent among
the different roles and interference between different loops or
conditionals is avoided.

There is a trade-off between efficiency and ease of programming that concerns how to ensure 
that all
the roles are aware of the evolution of the computation. Indeed, this can be
done in three ways: by using auxiliary communications generated either \emph{i})
by the projection (e.g., as for if and while constructs above) or \emph{ii}) by 
the semantics (as we will show for scopes) or \emph{iii}) by restricting the class of allowed \AIOC{}s (as done
for sequential composition using connectedness for sequence).
For instance, auxiliary communications for the 
$\ifthen{b \at r}{\mathcal{I}}{\mathcal{I}'}$ construct are needed unless one requires that $r \in \{ 
r_1,r_2\}$ for each $\pair{r_1}{r_2} \in \transI({\mathcal I}) \cup \transI({\mathcal I}')$. The use of  
auxiliary communications is possibly less efficient, while stricter connectedness 
conditions 
leave more burden on the shoulders of the programmer.

We now define the projection $\proj(\mathcal{I},\Sigma)$, based on the process-projection
$\pi$, to derive a \APOC{} network from a \AIOC{} process $\mathcal{I}$ and a global state
$\Sigma$.
We denote with $\parallel_{i \in I} \net_i$ the parallel composition of
networks $\net_i$ for each $i \in I$.
\begin{definition}[Projection]
The projection of a \AIOC{} process ${\mathcal I}$ with global state $\Sigma$ is
the \APOC{} network defined by $\proj({\mathcal I},\Sigma)=\parallel_{s \in
\roles({\mathcal I})} \role{\pi({\mathcal I},s),
\Sigma_s}{s}$
\end{definition}

\iftoggle{tech_report}{
Appendix \ref{app:es} shows the \APOC{} processes 
obtained by projecting the \AIOC{} for the Buying scenario on \buyer, \seller\,, and
\bank. 
}{
The technical report \cite{TR} shows the \APOC{} processes
obtained by projecting the \AIOC{} for the Buying scenario in Listing \ref{es:bank} on \buyer, 
\seller\,, 
and
\bank. 
} 

\subsection{\APOC{} semantics}
\begin{table*}[t]
{\footnotesize
\[
\begin{array}{c}
\mathax{One}{\role{\one,\Gamma}{r} \arro{\tick}
\role{\zero,\Gamma}{r}}
\hfill\quad
\mathrule{Assign}{\grass{e}_\Gamma = v}{\role{x = e,\Gamma}{r} \arro{\tau} 
\role{\one,\Gamma\substAPOC{v}{x}}{r}}
\hfill\quad
\mathax{Out-\AdaptRule}{\role{\cout{o^?}{X}{r'},\Gamma}{r} \arro{\coutLabel{o^?}{X}{r'}:r} 
\role{\one,\Gamma}{r}}
\\[.4cm]
\mathax{In}{\role{\cinp{o^?}{x}{r'},\Gamma}{r} \arro{\cinpLabel{o^?}{x}{v}{r'}:r}\role{x = 
v,\Gamma}{r}}\hfill
\mathrule{Out}{\grass{e}_\Gamma = v}{\role{\cout{o^?}{e}{r'},\Gamma}{r} 
\arro{\coutLabel{o^?}{v}{r'}:r} \role{\one,\Gamma}{r}}
\\[.7cm]
\mathrule{Sequence}{\role{P,\Gamma}{r} \arro{\delta} \role{P',\Gamma'}{r} \quad \delta \neq 
\tick}{\role{P \seqOp Q,\Gamma}{r} \arro{\delta} \role{P' \seqOp Q,\Gamma'}{r}}
\hfill
\mathrule{Seq-end}{\role{P,\Gamma}{r} \arro{\tick} \role{P',\Gamma}{r} \quad \role{Q,\Gamma}{r} 
\arro{\delta} \role{Q',\Gamma'}{r}}{\role{P \seqOp Q,\Gamma}{r} \arro{\delta} \role{Q',\Gamma'}{r}}
\iftoggle{tech_report}{\\[.7cm]}{\\}
\iftoggle{tech_report}{
\mathrule{Parallel}{\role{P,\Gamma}{r} \arro{\delta} \role{P',\Gamma'}{r} \quad \delta \neq 
\tick}{\role{P \mid Q,\Gamma}{r} \arro{\delta} \role{P' \mid Q,\Gamma'}{r}}
\hfill
\mathrule{Par-end}{\role{P,\Gamma}{r} \arro{\tick} \role{P',\Gamma}{r} \quad \role{Q,\Gamma}{r} 
\arro{\tick} \role{Q',\Gamma}{r}}{\role{P \mid Q,\Gamma}{r} \arro{\tick} \role{P' \mid 
Q',\Gamma}{r}} 
\\[.7cm]
\mathrule{If-then}{\grass{b}_\Gamma = \tt{true}}{\role{\ifthen{b}{P}{P'},\Gamma}{r} 
\arro{\tau} \role{P,\Gamma}{r}}
\hfill
\mathrule{If-else}{\grass{b}_\Gamma = \tt{false}}{\role{\ifthen{b}{P}{P'},\Gamma}{r}
\arro{\tau} \role{P',\Gamma}{r}}
\\[.7cm]
\mathrule{While-unfold}{\grass{b}_\Gamma = \tt{true}}{\role{\while{b}{P},\Gamma}{r} \arro{\tau} 
\role{P \seqOp \while{e}{P},\Gamma}{r}}
\hfill
\mathrule{While-exit}{\grass{b}_\Gamma = \tt{false}}{\role{\while{b}{P},\Gamma}{r} \arro{\tau} 
\role{\one,\Gamma}{r}}
\\
}{}
\\
\mathrule{Lead-\AdaptRule}{
\mathcal{I}' = 
\freshKey({\mathcal I}, n)
\qquad
\roles(\mathcal{I}') \subseteq S
}
{
\begin{array}{l}
\role{\pscope{n}{l}{r}{P}{\Delta}{S},\Gamma}{r} \arro{{\mathcal I}} 
\\ \quad \quad
\role{\prod_{r_i \in S \setminus \{r\}}
\cout{o^*_{n}}{\pi(\mathcal{I}',r_i)}{r_i}\seqOp
\pi(\mathcal{I}',r)\seqOp
\prod_{r_i \in S \setminus \{r\}} \cinp{o^*_{n}}{\_}{r_i},\Gamma}{r}\\
\end{array}
}
\\[1cm]
\mathax{Lead-\NoAdaptRule}{
\begin{array}{l}
\role{\pscope{n}{l}{r}{P}{\Delta}{S},\Gamma}{r}
\arro{\texttt{\NoAdaptLabel}} 
\\ \quad
\role{
\prod_{r_i \in S \setminus \{r\}} 
\cout{o^*_{n}}{\mbox{no}}{r_i}\seqOp P \seqOp \prod_{r_i \in S \setminus \{r\}} 
\cinp{o^*_{n}}{\_}{r_i},\Gamma}{r}\\
\end{array}
}
\\[.6cm]
\mathax{\AdaptRule}{\role{\psscope{n}{l}{r'}{P},\Gamma}{r} 
\arro{\cinpLabel{o^*_{n}}{\_}{P'}{r'}}
\role{P' \seqOp
\cout{o^*_{n}}{\texttt{ok}}{r'},\Gamma}{r}}
\\[.5cm]
\mathax{\NoAdaptRule}{\role{\psscope{n}{l}{r'}{P},\Gamma}{r} 
\arro{\cinpLabel{o^*_{n}}{\_}{\mbox{no}}{r'}}
\role{P \seqOp \cout{o^*_{n}}{\texttt{ok}}{r'},\Gamma}{r} }
\end{array}
\]
}
\caption{\iftoggle{tech_report}{\APOC{} role semantics.}{\APOC{} role 
semantics (excerpt).}}\label{table:apoc-proc}
\vspace{-1em}
\end{table*}

\begin{table*}[t]
{
\footnotesize
{
\[
\begin{array}{c}
\mathrule{Lift}{\net \arro{\delta} \net' 
\quad \delta \neq {\mathcal I}}{\tuple{\rules, \net} 
\arro{\delta} \tuple{\rules, \net'}}
\qquad\hfill
\mathrule{Lift-\AdaptRule}{\net \arro{\mathcal{I}} \net' \quad \mathcal{I} \textrm{ connected} \quad \mathcal{I} \in \rules}
{\tuple{\rules, \net} 
\arro{\mathcal{I}} \tuple{\rules,\net'}}
\hfill\qquad
\mathax{Change-Updates}{\tuple{\rules,\net} \arro{\rules'} 
\tuple{\rules',\net}}
\\
\mathrule{Synch}{\tuple{\rules, \net} \arro{\coutLabel{o^?}{v}{r_2}:r_1} \tuple{\rules, 
\net'} 
\quad \tuple{\rules,\net''} 
\arro{\cinpLabel{o^?}{x}{v}{r_1}:r_2} \tuple{\rules,\net'''}}
{\tuple{\rules, \net \parallel \net''} 
\arro{\commLabel{o^?}{r_1}{v}{r_2}{x} } \tuple{\rules, \net' \parallel \net'''}}
\\
\mathrule{Synch-\AdaptRule}{\tuple{\rules,\net}\arro{\coutLabel{o^?}{X}{r_2}:r_1}\tuple{\rules,
\net'} \quad \tuple{\rules,\net''} 
\arro{\cinpLabel{o^?}{\_}{X}{r_1}:r_2}\tuple{\rules,\net'''}}
{\tuple{\rules,\net \parallel \net''}
\arro{\commLabel{o^?}{r_1}{X}{r_2}{\_}}
\tuple{\rules,\net' \parallel \net'''}}
\\
\mathrule{Ext-Parallel}{\tuple{\rules,\net} \arro{\eta} \tuple{\rules,\net'} \qquad \eta \neq 
\tick}
{\tuple{\rules,\net\parallel \net''} 
\arro{\eta} \tuple{\rules,\net' \parallel \net''}}
\qquad\quad\hfill
\mathrule{Ext-Par-End}{\tuple{\rules,\net} \arro{\tick} \tuple{\rules,\net'} \quad 
\tuple{\rules,\net''} \arro{\tick} \tuple{\rules,\net'''}}{\tuple{\rules,\net\parallel 
\net''} \arro{\tick} \tuple{\rules,\net' \parallel \net'''}}
\end{array}
\]
}
}
\caption{\APOC{} system semantics.}
\label{table:apoc-sys}
\label{table:apoc-roles}
\label{table:apoc-net}
\end{table*}

\iftoggle{tech_report}{
\begin{definition}[\APOC{} systems semantics]
The semantics of \APOC{} systems is defined as the smallest LTS
closed under the rules in Tables~\ref{table:apoc-proc} and
\ref{table:apoc-sys}. Symmetric rules for parallel composition have
been omitted.
\end{definition}
} 
{ 
\begin{definition}[\APOC{} systems semantics]
The semantics of \APOC{} systems is defined as the smallest LTS closed under
the rules in Table \ref{table:apoc-sys} here and Table 7 in the companion technical report 
\cite{TR} (excerpt in
Table~\ref{table:apoc-proc}). 
Symmetric rules for parallel composition have
been omitted.
\end{definition}
} 

We use 
$\delta$ to
range over labels.
The semantics in the early style.
\iftoggle{tech_report}{
Rule \ruleName{In} receives a value $v$ from role $r'$ and assigns it to local
variable $x$ of $r$. Rules \ruleName{Out} and \ruleName{Out-\AdaptRule} execute
send and higher-order send actions, respectively. The send actions evaluate
expression $e$ in the local state $\Gamma$.
Rule \ruleName{One} terminates an empty process.
Rule \ruleName{Assign} executes an assignment ($\substAPOC{v}{x}$ represents the
substitution of value $v$ for variable $x$).
Rules \ruleName{Sequence} and  \ruleName{Seq-end} handle sequential
composition.
Rules \ruleName{Parallel} and  \ruleName{Par-end} handle the execution of
parallel processes.
Rules \ruleName{If-then} and \ruleName{If-else} execute the then or the else
branch in a conditional, respectively.
Rules \ruleName{While-unfold} and \ruleName{While-exit} model the unfolding or
the termination of a loop.

The other rules deal with code updates.
}
{
We comment below on the main rules.

Rule \ruleName{In} receives a value $v$ from role $r'$ and assigns it to local
variable $x$ of $r$. Rules \ruleName{Out} and \ruleName{Out-\AdaptRule} execute
send and higher-order send actions, respectively. The send evaluates
expression $e$ in the local state $\Gamma$. In rule \ruleName{Assign},
$\substAPOC{v}{x}$ represents the substitution of value $v$ for variable $x$.
}

Rule \ruleName{Lead-\AdaptRule} concerns the role $r$ coordinating the update
of a scope. Role $r$ decides which update to use. It is important that this
decision is taken by the unique coordinator $r$ for two reasons. First, $r$
ensures that all involved roles agree on whether to update or not. Second,
since the set of updates may change at any time, the choice of the update
inside $\rules$ needs to be atomic, and this is guaranteed using a unique
coordinator.
Role $r$ transforms the \AIOC{} ${\mathcal I}$ into ${\mathcal I}'$ using
function $\freshKey(\mathcal{I}, n)$, which produces a copy $\mathcal{I}'$ of
$\mathcal{I}$. In $\mathcal{I}'$ the indexes of scopes are fresh, which avoids
clashes with indexes already present in the target \APOC{}. Moreover, to avoid
that interactions in the update interfere with (parallel) interactions in the
context, $\freshKey(\mathcal{I}, n)$ renames all the operations inside
$\mathcal{I}$ by adding to them the index $n$. To this end we extend the set
of operations without changing the semantics. For each operation $o^?$ we
define extended operations of the form $\freshops{n}{o^?}$.
%
The coordinator $r$ also generates
the processes to be executed by the roles in $S$ using the process-projection
function $\pi$.  The processes are sent via higher-order communications only to the
roles that have to execute them. Then, $r$ starts its own updated code
$\pi(\mathcal{I}',r)$. Finally, auxiliary communications are used to
synchronise the end of the execution of the replaced process (here $\_$
denotes a fresh variable to store the synchronisation message
$\texttt{ok}$). The auxiliary communications are needed to ensure that the
update is performed in a coordinated way, i.e., the roles agree on when the
scope starts and terminates and on whether the update is performed
or not.

Rule \ruleName{Lead-\NoAdaptRule} instead defines the behaviour when the
coordinator $r$ decides to not update. In this case, $r$ sends
a token $\tt{no}$ to each other involved role, notifying
them that no update is applied.  End of scope synchronisation is as
above. Rules \ruleName{\AdaptRule} and \ruleName{\NoAdaptRule} define the behaviour of
the scopes for the other roles involved in the update. The scope
waits for a message from the coordinator. If the content of the message is
$\texttt{no}$, the body of the scope is executed. Otherwise, it is a process
$P'$ which is executed instead of the body of the scope.

Table~\ref{table:apoc-sys} defines the semantics of \APOC{} systems. We use
$\eta$ to range over \APOC{} systems labels. Rule \ruleName{Lift} and
\ruleName{Lift-\AdaptRule} lift roles transitions to the system level.
\ruleName{Lift-\AdaptRule} also checks that the update ${\mathcal I}$ is connected.
Rule \ruleName{Synch} synchronises a send with the corresponding receive,
producing an interaction.
Rule \ruleName{Synch-\AdaptRule} is similar, but it deals with
higher-order interactions. The labels of these transitions store the
information on the occurred communication: label
$\commLabel{o^?}{r_1}{v}{r_2}{x}$ denotes an interaction on operation
$o^?$ from role $r_1$ to role $r_2$ where the value $v$ is sent by
$r_1$ and then stored by $r_2$ in variable $x$. Label
$\commLabel{o^?}{r_1}{X}{r_2}{\_}$ denotes a similar interaction, but
concerning a higher-order value $X$. No receiver variable is
specified, since the received value becomes part of the code of
the receiving process. 
Rule \ruleName{Ext-Parallel} allows a network inside a parallel composition to
compute. Rule \ruleName{Ext-Par-End} synchronises the termination of parallel
networks. Finally, rule
\ruleName{Change-Updates} allows the set of updates to change
arbitrarily.

We can now define \APOC{} traces.
%
\begin{definition}[\APOC{} traces] 
A (strong) trace of a \APOC{} system $\tuple{\rules_1, \net_1}$ is a sequence (finite or 
infinite) of labels $\eta_1, \eta_2, \dots$ with $\eta_i \in \{\tau, 
\commLabel{o^?}{r_1}{v}{r_2}{x}, \tick, {\mathcal I},\texttt{\NoAdaptLabel}, 
\rules \}$  
such that there is a sequence of transitions \\
$\tuple{\rules_1,\net_1} \arro{\eta_1} \
\tuple{\rules_{2},\net_{2}} \arro{\eta_2} \dots$.\\
A weak trace of a \APOC{} system $\tuple{\rules_1,\net_1}$ is a sequence of labels $\eta_1, 
\eta_2, \dots$ obtained by removing all the labels corresponding to private communications, i.e.\ of 
the form $\commLabel{o^*}{r_1}{v}{r_2}{x}$ or $\commLabel{o^*}{r_1}{X}{r_2}{\_}$, and the silent 
labels 
$\tau$, from a trace of $\tuple{\rules_1,\net_1}$. Furthermore, all the extended operations 
of the form $\freshops{n}{o^?}$ are replaced by $o^?$.  
\end{definition}

Note that \APOC{} traces do not include send and receive actions. We do this
since these actions have no correspondence at the \AIOC{} level, where only 
whole interactions are allowed.

\iftoggle{tech_report}{
Note also that, in general, \APOC{}s can deadlock, e.g. $\role{\cinp{o}{x}{r'},\Gamma}{r}$ is a 
deadlocked 
\APOC{} network since all its traces contain only actions involving the change of the updates 
(i.e., 
labels $\rules$).
}{ 
}

Appendix~\ref{app:running_example} shows a sample
execution of the \APOC{} obtained by projecting the \AIOC{} for the Buying
scenario in Listing \ref{es:bank}.

\section{Correctness}\label{sec:corr}

In the previous sections we have presented \AIOC{}s, \APOC{}s, and described how
to derive a \APOC{} from a given \AIOC{}. This section presents the main
technical result of the paper, namely the correctness of the
projection. Correctness here means that the weak traces of a connected \AIOC{}
coincide with the weak traces of the projected \APOC{}.  

%
\begin{definition}[Trace equivalence]
 A \AIOC{} system $\tuple{\state, \rules,{\mathcal I}}$ and a \APOC{}
 system $\tuple{ \rules,\net}$ are (weak) trace equivalent iff
 their sets of (weak) traces coincide.
\end{definition}
\begin{theorem}[Correctness]\label{teo:final}
For each initial, connected \AIOC{} process ${\mathcal I}$, each state
$\state$, each set of updates
$\rules$, the \AIOC{} system $\tuple{\state, \rules,{\mathcal I}}$
and the \APOC{} system $\tuple{\rules,\proj({\mathcal I},\state)}$
are weak trace equivalent.
\end{theorem}

\iftoggle{tech_report}
{The proof of the theorem is reported in Appendix \ref{sec:proof}.
}
{
}

\iftoggle{tech_report}{
Trace-based properties of the \AIOC{} are inherited by the \APOC{}. Examples  include 
deadlock-freedom and 
termination.
}{
Trace-based properties of the \AIOC{} are inherited by the \APOC{}. Examples  include 
termination (see \cite{TR}) and deadlock-freedom.
}
\iftoggle{tech_report}
{
\begin{definition}[Deadlock-freedom and termination]
An internal \AIOC{} (resp.\ \APOC{}) trace is obtained by removing transitions labelled $\rules$ 
from a \AIOC{}  (resp.\ \APOC{}) trace. A \AIOC{} (resp.\ \APOC{}) system is 
deadlock-free if all its maximal finite internal traces have $\tick$ as label of the last 
transition. A \AIOC{} (resp.\ \APOC{}) system terminates if all its internal traces are finite. 
\end{definition}
}
{
\begin{definition}[Deadlock-freedom]
An internal \AIOC{} (resp.\ \APOC{}) trace is obtained by removing transitions labelled $\rules$ 
from a \AIOC{}  (resp.\ \APOC{}) trace. A \AIOC{} (resp.\ \APOC{}) system is 
deadlock-free if all its maximal finite internal traces have $\tick$ as label of the last 
transition. 
\end{definition}
}
Intuitively, internal traces are needed since labels $\rules$ do not
correspond to activities of the application and may be executed also
after application termination.

By construction initial \AIOC{}s are deadlock-free. Hence:
\begin{corollary}[Deadlock-freedom]\label{cor1}
For each initial, connected \AIOC{} ${\mathcal I}$, state $\Sigma$, and
set of updates $\rules$ the \APOC{} system $\tuple{\rules,\proj({\mathcal
    I},\Sigma)}$ is deadlock-free.
\end{corollary}
\iftoggle{tech_report}
{The proof of the corollary is reported in Appendix \ref{sec:proof}.
}
{
}
\iftoggle{tech_report}
{
\APOC{}s inherit termination from terminating \AIOC{}s.
\begin{corollary}[Termination]\label{cor2}
If the \AIOC{} system $\tuple{\Sigma,\rules,\mathcal{I}}$ terminates and 
$\mathcal{I}$ is connected then the \APOC{} system $\tuple{\rules,\proj({\mathcal I},\Sigma)}$ 
terminates.
\end{corollary}
\begin{proof}
It follows from the fact that only a finite number of auxiliary actions are added when moving from 
\AIOC{}s to \APOC{}s.
\end{proof}

Note that with arbitrary sets of updates no application may terminate. Hence, one has to 
restrict the allowed updates.
}

Moreover, our \AIOC{}s and \APOC{}s are free from races and orphan
messages. A race occurs when the same receive (resp.\ send) may
interact with different sends (resp.\ receives). In our setting, an
orphan message is an enabled send that is never consumed by a
receive. Orphan messages are more relevant in asynchronous systems,
where a message may be sent, and stay forever in the network, since
the corresponding receive operation may never become enabled.
However, even in synchronous systems orphan messages should be
avoided: the message is not communicated since the receive is not
available, hence a desired behaviour of the application never takes
place due to synchronization problems.

Trivially, \AIOC{}s avoid races and orphan messages since send and receive are
bound together in the same construct. Differently, at the \APOC{} level, since
all receive of the form $\cinp{o^?}{x}{r_1}$ in role $r_2$ may interact with
the sends of the form $\cout{o^?}{e}{r_2}$ in role $r_1$, races may happen.
However, thanks to the correctness of the projection, race-freedom holds also
for the projected \APOC{}s.

\begin{corollary}[Race-freedom]\label{cor:race}
For each initial, connected \AIOC{} ${\mathcal I}$, state $\Sigma$, and
set of updates $\rules$, if $\tuple{\rules,\proj({\mathcal I},\Sigma)}
\arro{\mu_1}\cdots\arro{\mu_n} \tuple{\rules',\net}$, then in $\net$ two
sends (resp.\ receives) cannot interact with the same receive (resp.\ send).
\end{corollary}

As far as orphan messages are concerned, they may appear in infinite
\APOC{} computations since a receive may not become enabled due to an
infinite loop. However, as a corollary of trace equivalence, we have
that terminating \APOC{}s are orphan message-free.

\begin{corollary}[Orphan message-freedom]\label{cor:orphan}
 For each initial, connected \AIOC{} ${\mathcal I}$, state $\Sigma$, and
set of updates $\rules$, if $\tuple{\rules,\proj({\mathcal I},\Sigma)}
 \arro{\mu_1}\cdots\arro{\tick} \tuple{\rules',\net}$, then $\net$ contains no sends.
\end{corollary}

\section{Related works and discussion}\label{sec:related} This paper presents
an approach for the dynamic update of distributed applications. Its
distinctive trait is that it guarantees the absence of communication deadlocks
and races by construction for the running distributed application, even in
presence of updates that were unknown when the application was started. More
generally, the \APOC{} is compliant with the \AIOC{} description, and inherits
its properties.
%

%
%

The two approaches closest to ours we are aware of are in the area of
multiparty session
types~\cite{hondaESOPext,hondaPOPL,POPLmontesi,Castagna}, and deal
with dynamic software updates~\cite{DSUtypes} and with monitoring of
self-adaptive systems~\cite{Dezani}.
The main difference between~\cite{DSUtypes} and our approach is that
\cite{DSUtypes} targets concurrent applications which are not
distributed.  Indeed, it relies on a check on the global state of the
application to ensure that the update is safe. Such a check cannot be
done by a single role, thus is impractical in a distributed setting.
Furthermore, the language in
\cite{DSUtypes} is much more constrained than ours, e.g., requiring
each pair of participants to interact on a dedicated pair of channels,
and assuming that all the roles not involved in a choice behave the
same in the two branches.
%
The approach in \cite{Dezani} is very different from ours, too. In particular,
in \cite{Dezani} all the possible behaviours are available since the very
beginning, both at the level of types and of processes, and a fixed adaptation
function is used to switch between them. This difference derives from the
distinction between self-adaptive applications, as they discuss, and
applications updated from the outside, as in our case.

We also recall \cite{DiGiustoP13}, which uses types to ensure safe
adaptation. However, \cite{DiGiustoP13} allows updates only when no
session is active, while we change the behaviour of running \AIOC{}s.

Our work is also similar to~\cite{MontesiCompChor}, which deals with
compositionality inside multiparty session types. However,
\cite{MontesiCompChor} only allows static parallel composition, while we
replace a term inside an arbitrary context at runtime. 
%

\iftoggle{tech_report}{
Extensions of multiparty session types with error
handling~\cite{carboneEXC,giachinoESCAPE} share with us the difficulties in
coordinating the transition from the expected pattern to an alternative
pattern, but in their case the error recovery pattern is known since the very
beginning, thus considerably simplifying the analysis. 

We briefly compare now
with works that exploit choreographic descriptions for adaptation, but with
very different aims.  For instance, \cite{ICSOC07} defines rules for adapting
the specification of the initial requirements for a choreography, thus keeping
the requirements up-to-date in presence of run-time changes. Our approach is
in the opposite direction: we are not interested in updating the system
specification tracking system updates, but in programming and ensuring
correctness of adaptation itself.

Other formal approaches to adaptation represent choreographies as
annotated finite state automata. In \cite{DYCHOR06} choreographies are used
to propagate protocol changes to the other peers, while \cite{SCC09}
presents a test to check whether a set of peers obtained from a
choreography can be reconfigured to match a second
one. Differently from ours, these works only provide change
recommendations for adding and removing message sequences.
}{ 
}


%
%

In principle, our update mechanism can be used to inject guarantees of
freedom from deadlocks and races into existing approaches to
adaptation, e.g., the ones in the surveys~\cite{adapt-survey12,ADAPTGhezzi}.
However, this task is cumbersome, due to the
huge number and heterogeneity of those approaches, and since for each of them
the integration with our techniques is far from trivial. Nevertheless, we
already started it. Indeed, in~\cite{SLE}, we apply our technique to the
approach described in~\cite{JORBApaper}.
While applications in~\cite{JORBApaper} are not distributed and there are no
guarantees on the correctness of the application after adaptation,
applications in~\cite{SLE}, based on the same adaptation mechanisms, are
distributed and free from deadlocks and races by construction.
%

Furthermore, on the website~\cite{AIOCJ}, we give examples of how to integrate
our approach with distributed~\cite{distributedAOP} and
dynamic~\cite{YangCSSSM02} Aspect-Oriented Programming (AOP) and with
Context-Oriented Programming (COP)
\cite{cop}. In general, we can deal with cross-cutting concerns like logging
and authentication, typical of AOP, viewing pointcuts as empty scopes and
advices as updates. Layers, typical of COP, can instead be defined by updates
which can fire according to contextual conditions.
We are also planning to apply our techniques to multiparty session types
\cite{hondaESOPext,hondaPOPL,POPLmontesi,Castagna}. The main challenge here is
to deal with multiple interleaved sessions. An initial analysis of the problem
is presented in~\cite{WS-SEFM}.

\bibliographystyle{abbrv}
\bibliography{biblio}

\begin{thebibliography}{10}

\bibitem{AIOCJ}
{AIOCJ website}.
\newblock \url{http://www.cs.unibo.it/projects/jolie/aiocj.html}.

\bibitem{DSUtypes}
G.~Anderson and J.~Rathke.
\newblock Dynamic software update for message passing programs.
\newblock In {\em APLAS}, volume 7705 of {\em LNCS}, pages 207--222. Springer,
  2012.

\bibitem{WS-SEFM}
M.~Bravetti et~al.
\newblock Towards global and local types for adaptation.
\newblock In {\em SEFM Workshops}, volume 8368 of {\em LNCS}, pages 3--14.
  Springer, 2013.

\bibitem{giachinoESCAPE}
S.~Capecchi, E.~Giachino, and N.~Yoshida.
\newblock {Global Escape in Multiparty Sessions}.
\newblock In {\em Proc. of FSTTCS 2010}, volume~8 of {\em LIPIcs}, pages
  338--351. Schloss Dagstuhl, 2010.

\bibitem{carboneEXC}
M.~Carbone, K.~Honda, and N.~Yoshida.
\newblock {Structured Interactional Exceptions in Session Types}.
\newblock In {\em Proc. of CONCUR'08}, volume 5201 of {\em LNCS}, pages
  402--417. Springer, 2008.

\bibitem{hondaESOPext}
M.~Carbone, K.~Honda, and N.~Yoshida.
\newblock Structured communication-centered programming for web services.
\newblock {\em ACM Trans. Program. Lang. Syst.}, 34(2):8, 2012.

\bibitem{POPLmontesi}
M.~Carbone and F.~Montesi.
\newblock {Deadlock-Freedom-by-Design: Multiparty Asynchronous Global
  Programming}.
\newblock In {\em POPL}, pages 263--274. ACM, 2013.

\bibitem{Castagna}
G.~Castagna, M.~Dezani-Ciancaglini, and L.~Padovani.
\newblock On global types and multi-party session.
\newblock {\em Logical Methods in Computer Science}, 8(1), 2012.

\bibitem{Dezani}
M.~Coppo, M.~Dezani-Ciancaglini, and B.~Venneri.
\newblock Self-adaptive monitors for multiparty sessions.
\newblock In {\em PDP}, pages 688--696. IEEE, 2014.

\bibitem{SLE}
M.~{Dalla Preda}, S.~Giallorenzo, I.~Lanese, J.~Mauro, and M.~Gabbrielli.
\newblock {AIOCJ:} {A} choreographic framework for safe adaptive distributed
  applications.
\newblock In {\em SLE}, volume 8706 of {\em LNCS}, pages 161--170. Springer,
  2014.

\bibitem{TR}
M.~{Dalla Preda}, I.~Lanese, J.~Mauro, M.~Gabbrielli, and S.~Giallorenzo.
\newblock {Dynamic Choreographies: Safe Runtime Updates of Distributed
  Applications}.
\newblock \url{http://www.cs.unibo.it/projects/jolie/dioc.pdf}.

\bibitem{DiGiustoP13}
C.~{Di Giusto} and J.~A. P{\'{e}}rez.
\newblock Disciplined structured communications with consistent runtime
  adaptation.
\newblock In {\em SAC}, pages 1913--1918. ACM, 2013.

\bibitem{ADAPTGhezzi}
C.~Ghezzi, M.~Pradella, and G.~Salvaneschi.
\newblock An evaluation of the adaptation capabilities in programming
  languages.
\newblock In {\em SEAMS}, pages 50--59. ACM, 2011.

\bibitem{cop}
R.~Hirschfeld, P.~Costanza, and O.~Nierstrasz.
\newblock {Context-oriented Programming}.
\newblock {\em Journal of Object Technology}, 7(3):125--151, 2008.

\bibitem{hondaPOPL}
K.~Honda, N.~Yoshida, and M.~Carbone.
\newblock {Multiparty Asynchronous Session Types}.
\newblock In {\em POPL}, pages 273--284. ACM Press, 2008.

\bibitem{ICSOC07}
I.~Jureta, S.~Faulkner, and P.~Thiran.
\newblock Dynamic requirements specification for adaptable and open
  service-oriented systems.
\newblock In {\em ICSOC}, volume 4749 of {\em LNCS}, pages 270--282. Springer,
  2007.

\bibitem{JORBApaper}
I.~Lanese, A.~Bucchiarone, and F.~Montesi.
\newblock {A Framework for Rule-Based Dynamic Adaptation}.
\newblock In {\em TGC}, volume 6084 of {\em LNCS}, pages 284--300. Springer,
  2010.

\bibitem{SEFM08}
I.~Lanese, C.~Guidi, F.~Montesi, and G.~Zavattaro.
\newblock {Bridging the Gap between Interaction- and Process-Oriented
  Choreographies}.
\newblock In {\em SEFM}, pages 323--332. IEEE Press, 2008.

\bibitem{WWVlanese}
I.~Lanese, F.~Montesi, and G.~Zavattaro.
\newblock Amending choreographies.
\newblock In {\em WWV}, volume 123, pages 34--48. EPTCS, 2013.

\bibitem{adapt-survey12}
L.~A.~F. Leite et~al.
\newblock A systematic literature review of service choreography adaptation.
\newblock {\em Service Oriented Computing and Applications}, 7(3):199--216,
  2013.

\bibitem{MontesiCompChor}
F.~Montesi and N.~Yoshida.
\newblock Compositional choreographies.
\newblock In {\em CONCUR}, volume 8052 of {\em LNCS}, pages 425--439. Springer,
  2013.

\bibitem{scoop}
P.~Nienaltowski.
\newblock {\em Practical framework for contract-based concurrent
  object-oriented programming}.
\newblock PhD thesis, ETH Zurich, 2007.

\bibitem{distributedAOP}
R.~Pawlak et~al.
\newblock {JAC: an aspect-based distributed dynamic framework}.
\newblock {\em Softw., Pract. Exper.}, 34(12):1119--1148, 2004.

\bibitem{DYCHOR06}
S.~Rinderle, A.~Wombacher, and M.~Reichert.
\newblock Evolution of process choreographies in dychor.
\newblock In {\em OTM Conferences (1)}, volume 4275 of {\em LNCS}, pages
  273--290. Springer, 2006.

\bibitem{rust}
Rust website.
\newblock \url{http://www.rust-lang.org/}.

\bibitem{scribble}
Scribble website.
\newblock \url{http://www.jboss.org/scribble}.

\bibitem{SCC09}
A.~Wombacher.
\newblock Alignment of choreography changes in {BPEL} processes.
\newblock In {\em IEEE SCC}, pages 1--8. IEEE Press, 2009.

\bibitem{YangCSSSM02}
Z.~Yang, B.~H.~C. Cheng, R.~E.~K. Stirewalt, J.~Sowell, S.~M. Sadjadi, and
  P.~K. McKinley.
\newblock An aspect-oriented approach to dynamic adaptation.
\newblock In {\em WOSS}, pages 85--92. ACM, 2002.

\end{thebibliography}

\newpage
\appendix


\iftoggle{tech_report}{

\section{Projecting the \AIOC{} for the Buying scenario on \buyer, \seller, and \bank. }
\label{app:es}

This section shows the projections of the \AIOC{} process defined in Listing \ref{es:bank} on the 
\bank, \buyer, and \seller\ roles.

In order to define the projection we first have to annotate the \AIOC{}. This leads to the 
following annotated \AIOC{}.

\begin{minipage}{\columnwidth}
  \begin{lstlisting}[mathescape=true,numbersep=1em,xleftmargin=2em,caption=Annotated 
\AIOC { } process for Buying Scenario.]
1 : price_ok@buyer = false;
2 : continue@buyer = true;
3 : while ( !price_ok and continue )@buyer {
 4 : b_prod@buyer = getInput();
 5 : priceReq : buyer( b_prod ) $\rightarrow$ seller( s_prod );
 6 : scope @seller {
   7 : s_price@seller = getPrice( s_prod );
   8 : offer : seller( s_price ) $\rightarrow$ buyer( b_price )
 };
 9: price_ok@buyer = getInput();
 10 : if ( !price_ok )@buyer {
  11 : continue@buyer = getInput()} };
12 : if ( price_ok )@buyer {
 13 : payReq : seller( payDesc( s_price ) ) $\rightarrow$ bank( desc );
 14 : scope @bank {
   15 : payment_ok@bank = true;
   16 : pay : buyer( payAuth( b_price ) ) $\rightarrow$ bank( auth );
   ... // code for the payment
 };
 17 : if ( payment_ok )@bank {
   18 : confirm : bank( null ) $\rightarrow$ seller( _ ) |
   19 : confirm : bank( null ) $\rightarrow$ buyer( _ )
 } else {
   20 : abort : bank( null ) $\rightarrow$ buyer( _ )
 } }
\end{lstlisting}
\end{minipage}

We are ready to compute the projection on the 
\bank, \buyer, and \seller\ roles respectively.
To improve readability, we omit some 
$\one$
processes that have no impact on the behaviour.

\begin{minipage}{\columnwidth}
\begin{lstlisting}[mathescape=true,caption={Bank \APOC{} Process},label=lst:use_caseAPOCbank]
o$^*_{12}$ : x$_{12}$ from buyer;
if ( x$_{12}$ ) {
  payReq : desc from seller;
  $14$ : scope @bank {
        payment_ok = true;
        pay : auth from buyer;
        ... // code for the payment
        }
      roles { buyer, bank };
  if ( payment_ok ) { 
    { o$^*_{17}$ : true to seller | o$^*_{17}$ : true to buyer };
    { confirm : null to seller | confirm : null to buyer }
  } else {
    { o$^*_{17}$ : true to seller | o$^*_{17}$ : true to buyer };
    abort : null to buyer } }
\end{lstlisting}
\end{minipage}

\begin{minipage}{\columnwidth}
\begin{lstlisting}[mathescape=true,caption={Buyer \APOC{} 
Process},label=lst:use_caseAPOCbuyer]
price_ok = false; continue = true;
while ( not( price_ok ) and continue ) {
  o$^*_{3}$ : true to seller;
  b_prod = getInput();
  priceReq : b_prod to seller;
  6 : scope @seller {      
      offer : b_price from seller }
  price_ok = getInput();
  if ( not( price_ok ) ) { continue = getInput() };
  o$^*_{3}$ : _ from seller };
o$^*_{3}$ : false to seller;
if ( price_ok ) {
  { o$^*_{12}$ : true to seller | o$^*_{12}$ : true to bank };
  14 : scope payment@bank {
    pay : payAuth( b_price ) to bank };
  o$^*_{17}$ : x$_{17}$ from bank;
  if ( x$_{17}$ ) { confirm : _ from bank
  } else { abort : _ from bank } }
\end{lstlisting}
\end{minipage}

\begin{minipage}{\columnwidth}
\begin{lstlisting}[mathescape=true,caption={Seller \APOC{} 
Process},label=lst:use_caseAPOCseller]
o$^*_{3}$ : x$_{3}$ from buyer;
while ( x$_{3}$ ) {
  priceReq : s_prod from buyer;
  6 : scope @seller {      
        s_price = getPrice( s_prod );
        offer : s_price to buyer }
      roles { seller, buyer };  
  o$^*_{3}$ : ok to buyer;
  o$^*_{3}$ : x$_{3}$ from buyer };
o$^*_{12}$ : x$_{12}$ from buyer;
if ( x$_{12}$ ) {
  payReq : payDesc( s_price ) to bank;
  o$^*_{17}$ : x$_{17}$ from bank;
  if ( x$_{17}$ ) { confirm : _ from bank } }
\end{lstlisting}
\end{minipage}

 \section{Running example of scope update}
\label{app:running_example}

This section shows an example of how updates are performed. We consider an excerpt of the 
choreography of the Buying Scenario (Listing \ref{es:bank}) simulating the update of the scope in 
Lines 5-8.
To this end, we assume that the seller direction decides to stimulate business by 
using the update in Listing \ref{rule_price_inquiry}.



Let us consider both the \AIOC{} and
the \APOC{} level, dropping some $\one$s to improve readability.
%
Assume that the \lstinline{buyer} has just sent the name of the product
(s)he is interested in to the \lstinline{seller} (Line 4) and consider the 
following annotated \AIOC{}:

\begin{lstlisting}[mathescape=true]
6 : scope @seller {
  7 : s_price@seller = getPrice( s_prod );
  8 : offer : seller( s_price ) ) $\rightarrow$ buyer( b_price )
}
\end{lstlisting}

At the \AIOC{} level,
the scope
\lstinline{price-inquiry} is atomically substituted with the new code with fresh indexes.  Then, 
the \AIOC{} reduces 
to:

\begin{lstlisting}[mathescape=true]
21 : cardReq : seller( null ) $\rightarrow$ buyer( _ );
22 : card_id@buyer = getInput();
23 : card : buyer( card_id ) $\rightarrow$ seller( buyer_id );
24 : if isValid( buyer_id )@seller {
  25 : s_price@seller = getPrice( s_prod ) * 0.9    
} else {
  26 : s_price@seller = getPrice( s_prod )
};
27 : offer : seller( s_price ) $\rightarrow$ buyer( b_price )
\end{lstlisting}

At the \APOC{} level, this operation is not atomic, since the
scope is distributed between two participants, and the
coordination protocol is explicitly represented.

To clarify this point, let us consider the \APOC{} process $P_b$ below,
obtained by projecting the \AIOC{} of the update in Listing~\ref{rule_price_inquiry} on 
the 
\lstinline{buyer} role.
\begin{lstlisting}[mathescape=true]
$P_{b}$ := cardReq : null from seller;
  card_id = getInput();
  card : card_id to seller;
  offer : b_price from seller
\end{lstlisting}
At the \APOC{} level, the first step of the update protocol is
performed by the \lstinline{seller}. The \APOC{} description of the
\lstinline{seller} before the update is:

\begin{lstlisting}[mathescape=true]
6 : scope @seller {
      s_price = getPrice( s_prod );
      offer : s_price to buyer }
      roles { seller, buyer } 
\end{lstlisting}
When the scope construct is enabled, the \lstinline{seller}, being the
coordinator of the update, decides to update using the code in Listing~\ref{rule_price_inquiry}. 
Thus, the \lstinline{seller} reduces to:
\begin{lstlisting}[mathescape=true]
o$^*_{6}$ : $P_{b}$ to buyer;
cardReq : null to buyer;
card : buyer_id from buyer;
if isValid( buyer_id ) {
  s_price = getPrice( s_prod ) * 0.9
  } else { s_price@seller = getPrice( s_prod ) };
offer : s_price to buyer;
o$^*_{6}$ : _ from buyer;
\end{lstlisting}
First, the \lstinline{seller} requires the \lstinline{buyer} to update, sending to him the new 
\APOC{} fragment to execute.
Then, the \lstinline{seller} starts to execute its own updated
\APOC{}. When the new \APOC{} code is terminated, (s)he waits for the
notification of the termination of the \APOC{} fragment executed by the
\lstinline{buyer}.

As far as the \lstinline{buyer} is concerned, the \APOC{} before
the update is as follows.
\begin{lstlisting}[mathescape=true]
6 : scope @seller {
  offer : s_price from seller
}
\end{lstlisting}
The scope construct in the \lstinline{buyer} waits for the arrival of
a message from the coordinator of the update. In case an update has to
be applied, this message contains the \APOC{} fragment to
execute. Once this message is received, the scope construct is
replaced by the received \APOC{} fragment, followed by the
notification of termination to the \lstinline{seller}.
\begin{lstlisting}[mathescape=true]
$P_{b}$ ; o$^*_{6}$ : ok to seller
\end{lstlisting}

Let us now consider the case where the application is not updated.  At
the \AIOC{} level, the scope construct simply disappears, and its body
becomes enabled.
\begin{lstlisting}[mathescape=true]
s_price@seller = getPrice( s_prod );
offer : seller( s_price ) ) $\rightarrow$ buyer( b_price )
\end{lstlisting}

As before, at the \APOC{} level this operation is not atomic.  In
particular, the \APOC{} process of the \lstinline{seller} becomes as
follows.

\begin{lstlisting}[mathescape=true]
o$^*_{6}$ : no to buyer;
s_price = getPrice( s_prod );
offer :  s_price to buyer;
o$^*_{6}$ : _ from buyer;
\end{lstlisting}

Here the \lstinline{seller} notifies to the \lstinline{buyer} that
no update is performed, and then proceeds with the normal
execution. Then, as before, (s)he waits for the notification of the
termination of the body of the scope from the \lstinline{buyer}.
Dually, the \lstinline{buyer} waits for the arrival of the message. If
the message states that no update is needed, the scope construct
is removed and its body executed. At the end, a notification of
termination is sent to the coordinator of the update:  

\begin{lstlisting}[mathescape=true]
offer :  b_price from seller;
o$^*_{6}$ : ok to seller;
\end{lstlisting}

 \section{Proof of Theorem~\ref{teo:compl}}
\label{app:complexity}

In order to prove the bound on the complexity of the connectedness
check we use the lemma below, showing that the checks to verify the
connectedness for sequence for a single sequence operator can be
performed in linear time on the size of the sets generated by
$\transI$ and $\transF$.

\begin{lemma}\label{lemma:limit_connectedness}
Given $S, S'$ sets of multisets of two elements, checking if $\forall s \in S
\ . \ \forall s'\in S' \ . \ s \cap s' \neq \emptyset$ can be done in $O(n)$
steps, where $n$ is the maximum of $|S|$ and $|S'|$.
\end{lemma}

\begin{proof}
W.l.o.g.~we can assume that $|S| \leq |S'|$. If $|S| \leq 9$ then the check
can be performed in $O(n)$ by comparing all the elements in $S$ with all the
elements in $S'$. If $|S| > 9$ then at least 4 distinct elements appear in the
multisets in $S$ since the maximum number of multisets with cardinality 2
obtained by $3$ distinct elements is $9$. In this case the following cases
cover all the possibilities:
\begin{itemize}
 \item there exist distinct elements $a,b,c,d$ s.t.~$\{a,b\},\{a,c\}$, and
$\{a,d\}$ belong to $S$. In this case for the check to succeed all the
multisets in $S'$ must contain $a$, otherwise the intersection of the multiset
not containing $a$ with one among the multisets $\{a,b\},\{a,c\}$, and
$\{a,d\}$ is empty. Similarly, since $|S'| > 9$, for the check to succeed all
the multisets in $S$ must contain $a$. Hence, if $\{a,b\},\{a,c\}$, and
$\{a,d\}$ belong to $S$ then the check succeeds iff $a$ belongs to all the
multisets in $S$ and in $S'$.
 \item there exist distinct elements $a,b,c,d$ s.t.~$\{a,b\}$ and $\{c,d\}$
belong to $S$. In this case the check succeeds only if $S'$ is a subset of $\{
\{a,c\},\{a,d\},\{b,c\},\{b,d\} \}$. Since $|S'| > 9$ the check can never
succeed.
 \item there exist distinct elements $a,b,c$ s.t.~$\{a,a\}$ and $\{b,c\}$
belong to $S$. In this case the check succeeds only if $S'$ is a subset of $\{
\{a,b\},\{a,c\} \}$. Since $|S'| > 9$ the check can never succeed.
 \item there exist distinct elements $a,b$ s.t.~$\{a,a\}$ and $\{b,b\}$ belong
to $S$. In this case the check succeeds only if $S'$ is a subset of $\{
\{a,b\} \}$. Since $|S'| > 9$ the check can never succeed.
\end{itemize}
Summarising, if $|S| > 9$ the check can succeed iff all the multisets in $S$
and in $S'$ share a common element. The existence of such an element can be
verified in time $O(n)$.
\end{proof}

\setcounter{theorem}{0}
\begin{theorem}[Connectedness-check complexity] \mbox{} \\
The connectedness of a \AIOC{} process ${\mathcal I}$ can be checked in
time $O(n^2 \log(n))$, where $n$ is the number of nodes in the abstract syntax
tree of ${\mathcal I}$.
\end{theorem}

\begin{proof}
To check the connectedness of ${\mathcal I}$ we first compute the
values of the functions $\transI$, $\transF$, and $\operations$ for
each node of the abstract syntax tree (AST). We then check for each
sequence operator whether connectedness for sequence holds and for each
parallel operator whether connectedness for parallel holds.

The functions $\transI$ and $\transF$ associate to each node a set of
pairs of roles.  Assuming an implementation of the data set structure
based on balanced trees (with pointers), $\transI$ and $\transF$ can
be computed in constant time for interactions, assignments, $\one$,
$\zero$, and sequence constructs. For while and scope constructs
computing $\transF({\mathcal I}')$ requires the creation of balanced
trees having an element for every role of ${\mathcal I}'$. Since the
roles are $O(n)$, $\transF({\mathcal I}')$ can be computed in
$O(n\log(n))$.  For parallel and if constructs a union of sets is
needed. The union costs $O(n\log(n))$ since each set generated by
$\transI$ and $\transF$ contains at maximum $n$ elements.

The
computation of $\operations$ can be performed in $O(1)$ except
for the parallel, sequence, and if constructs, where the union of sets
costs $O(n\log(n))$.
Since the AST contains $n$ nodes, the computation of the sets generated by
$\transI$, $\transF$, and $\operations$ can be performed in $O(n^2\log(n))$.

To check connectedness for sequence we have to verify that for each node
${\mathcal I}'\seqOp{\mathcal I}''$ of the AST $\forall
\pair{r_1}{r_2} \in \transF({\mathcal I}'), \forall
\pair{s_1}{s_2} \in \transI({\mathcal I}'') \; . \;
\{r_1,r_2\} \cap \{s_1,s_2\} \neq \emptyset$.
Since $\transF({\mathcal I}')$ and $\transI({\mathcal I}'')$ have $O(n)$
elements, thanks to Lemma \ref{lemma:limit_connectedness}, checking if
${\mathcal I}'\seqOp{\mathcal I}''$ is connected for sequence costs $O(n)$.
Since in the AST there are less than $n$ sequence operators, checking the
connectedness for sequence on the whole AST costs $O(n^2)$.

To check connectedness for parallel we have to verify that for each node
${\mathcal I}' \parOpI {\mathcal I}''$ of the AST we have that
$\operations({\mathcal I}') \cap \operations({\mathcal I}'') = \emptyset$.
Since $\operations({\mathcal I}')$ and $\operations({\mathcal I}'')$ have
$O(n)$ elements, checking if their intersection is empty costs $O(n\log(n))$.
Since in the AST there are less than $n$ parallel operators, checking the
connectedness for parallel on the whole AST costs $O(n^2\log(n))$.

The complexity of checking the connectedness of the entire AST is therefore
limited by the cost of computing functions $\transI$, $\transF$, and
$\operations$, and of checking the connectedness for parallel. All these
activities have a complexity of $O(n^2\log(n))$.
\end{proof}

\section{Proof of Theorem~\ref{teo:final}}\label{sec:proof}
This section presents the proof of our main result, Theorem~\ref{teo:final},
including various auxiliary definitions and lemmas.\\

The proof strategy consists in defining a notion of bisimilarity
(Definition~\ref{def:bisim}) which implies weak trace equivalence
(Lemma~\ref{lemma:bis2trsynch}) and then providing a suitable bisimulation
relating each well-annotated connected \AIOC{} system with its projection.
Such a relation is not trivial, since events which are atomic in the \AIOC{},
e.g., the evaluation of the guard of a conditional (including removing the
discarded branch), are no more atomic in the
\APOC{}. In the case of conditional, the \AIOC{} transition is mimicked by a
conditional performed by the role evaluating the guard, a set of auxiliary
communications sending the value of the guard to the other roles, and local
conditionals based on the received value.  These mismatches are taken care by
function $\upd$ (Definition~\ref{def:upd}). This function needs also to remove
the auxiliary communications allowing to synchronise the termination of
scopes, which have no counterpart after the \AIOC{} scope has been consumed.
However, we have to record their impact on the possible executions. Thus we
define an event structure for \AIOC{} (Definition~\ref{def:ev}) and one for
\APOC{} (Definition~\ref{def:apocev}) and we show that the two are related
(Lemma~\ref{lemma:ev}).\\

In the main part, we defined annotated \AIOC{}s
(Definition~\ref{def:annAIOC}). Here we also need to speak about their
semantics. Indeed, annotated \AIOC{}s trivially inherit the semantics of
\AIOC{}s, since indexes are just decorations, with no effect on the behaviour.
The only tricky points are in rule \ruleName{Interaction}, where the
assignment inherits the index from the interaction, in rule
\ruleName{While-unfold}, where the body is copied together with its indexes,
and in rule \ruleName{\AdaptRule}, where one has to ensure that indexes of
constructs from the body of the update are never used elsewhere in the
\AIOC{}.

Notably, due to while unfolding, uniqueness of indexes is not preserved by
transitions. To solve this problem we build global indexes on top of indexes.
Uniqueness of global indexes is preserved by transitions. The same
construction can be applied both at the \AIOC{} level and at the \APOC{}
level.


\begin{definition}[Global index] Given an annotated \AIOC{} process
$\mathcal{I}$, or an annotated \APOC{} network $\net$ (defined later on), for
each annotated construct with index $n$ we define its global index $\xi$ as
follows:
\begin{itemize}
\item if the construct is not in the body of a while then $\xi = n$;
\item if the innermost while construct that contains the considered construct
has global index $\xi'$ then the considered construct has global index $\xi =
\xi':n$.
\end{itemize}
\end{definition}

Using global indexes we can now define event structures corresponding to the
execution of \AIOC{}s and \APOC{}s. We start by defining \AIOC{} events. Some
events correspond to transitions of the \AIOC{}, and we say that they are
enabled when the corresponding transition is enabled, executed when the
corresponding transition is executed. \AIOC{} events are defined on annotated
\AIOC{}s. Note that a non-annotated \AIOC{} can always be annotated.

\begin{definition}[\AIOC{} events]\label{def:ev} We use $\ev$ to range over
events, and we write $[\ev]_r$ to highlight that event $\ev$ is performed by
role $r$.
An annotated \AIOC{} ${\mathcal I}$ contains the following events:

{\bf Communication events:} a sending event $\xi : \co{o^?}{r_2}$ in role
$r_1$ and a receiving event $\xi: \ci{o^?}{r_1}$ in role $r_2$ for each
interaction $n: \comm{o^?}{r_1}{e}{r_2}{x}$ with global index $\xi$; we also
denote the sending event as $f_\xi$ or $[f_\xi]_{r_1}$ and the receiving event
as $t_\xi$ or $[t_\xi]_{r_2}$. Sending and receiving events correspond to the
transition executing the interaction.

{\bf Assignment events:} an assignment event $\ev_\xi$ in role $r$ for each
assignment $n:\assign{x}{r}{e}$ with global index $\xi$; the event corresponds
to the transition executing the assignment.

{\bf Scope events:} a scope initialisation event $\uparrow_{\xi}$ and a scope
termination event $\downarrow_{\xi}$ for each scope
$n:\scope{l}{r}{\mathcal{I}}{\Delta}{A}$ with global index $\xi$. Both these
events belong to all the roles in $\roles(\mathcal{I})$. The scope
initialisation event corresponds to the transition performing or not
performing an update on the given scope. The scope termination event is just
an auxiliary event (related to the auxiliary interactions implementing the
scope termination).

{\bf If events:} a guard if-event $\ev_\xi$ in role $r$ for each construct
$\ifthenKey{b \at r}{\mathcal{I}}{\mathcal{I}'}{n}$ with global index $\xi$;
the guard-if event corresponds to the transition evaluating the guard of the
if.

{\bf While events:} a guard while-event $\ev_\xi$ in role $r$ for each
construct $\whileKey{b \at r}{\mathcal{I}}{n}$ with global index $\xi$; the
guard-while event corresponds to the transition evaluating the guard of the
while.

Function $\event({\mathcal I})$ denotes the set of events of the annotated
\AIOC{} ${\mathcal I}$. A sending and a receiving event with the same global
index $\xi$ are called matching events. We denote with $\overline{\ev}$ an
event matching event $\ev$.
\end{definition}

Note that there are events corresponding to just one execution of the while.
If unfolding is performed, new events are created.

The relation below defines a causality relation among events based on the
constraints given by the semantics on the execution of the corresponding
transitions.

\begin{definition}[\AIOC{} causality relation]\label{def:causalAIOC} Let us
consider an annotated \AIOC{} ${\mathcal I}$.  A causality relation $\leqaioc
\, \subseteq \, \devent({\mathcal I}) \times \devent({\mathcal I})$ is a
partial order among events in ${\mathcal I}$. We define $\leqaioc$ as the
minimum partial order satisfying:

{\bf Sequentiality:} let ${\mathcal I}' \seqOp {\mathcal I}''$ be a subterm of
\AIOC{} ${\mathcal I}$. If $\ev'$ is an event in ${\mathcal I}'$ and $\ev''$
is an event in ${\mathcal I}''$, then $\ev'
\leqaioc \ev''$.

{\bf Scope:} let $n:\adapt{{\mathcal I}'}{l}\Delta{r}$ be a subterm of
\AIOC{} ${\mathcal I}$. If $\ev'$ is an event in ${\mathcal I}'$ then
$\uparrow_{\xi} \leqaioc \ev' \leqaioc \downarrow_{\xi}$.

{\bf Synchronisation:} for each interaction the sending event precedes the
receiving event.

{\bf If:} let $\ifthenKey{b \at r}{\mathcal{I}}{\mathcal{I}'}{n}$ be a subterm
of \AIOC{} $\mathcal{I}$, let $\ev_\xi$ be the guard if-event in role $r$,
then for every event $\ev$ in ${\mathcal I}$ and for every event $\ev'$ in
${\mathcal I}'$ we have $\ev_\xi \leqaioc \ev$ and $\ev_\xi \leqaioc
\ev'$.

{\bf While:} let $\whileKey{b \at r}{\mathcal{I}}{n}$ be a subterm of \AIOC{}
$\mathcal{I}$,  let $\ev_\xi$ be the guard while-event in role $r$, then for
every event $\ev$ in ${\mathcal I}'$ we have $\ev_\xi
\leqaioc \ev$.
\end{definition}

We now define events and the corresponding causality relation also for
\APOC{}s.  First, we need to define annotated \APOC{}s. Annotations
for \APOC{}s exploit not only indexes $i \in \mathbb{N}$, but also
indexes of the form $(i,\mathit{true})$ and $(i,\mathit{false})$ with $i \in
\mathbb{N}$.  We use $n$ to range over all forms of indexes. 

%
\begin{definition}[Annotated \APOC{}] In \APOC{} networks, scopes are already
annotated.  Annotated \APOC{} networks are obtained by adding indexes
$n$ also to communication primitives, assignments, while, and if
constructs, thus obtaining the following grammar:
\begin{eqnarray*}
P & \gram & n: \cinp{o^?}{x}{r} \mid n:\cout{o^?}{e}{r} \mid
n:\cout{o^\ast}{X}{r} \mid\\ & & P\seqOp P'\mid \ P\mid P' \mid  n: x = e \mid
\one \mid \zero \mid\\ & & \ifthenKey{b}{P}{P'}{n} \mid
\whileKey{b}{P}{n}\mid\\ & & \pscope{n}{l}{r}{P}{\Delta}{S} \mid\\ & &
\psscope{n}{l}{r}{P}\\ X & \gram & {\tt{no}} \mid P\\
\net & \gram & \role{P,\Gamma}{r}\mid \net \parallel \net'
\end{eqnarray*}
\end{definition} 
We extend the projection function so to generate annotated \APOC{}
networks from annotated \AIOC{} processes. It requires that all the
\APOC{} constructs obtained projecting a \AIOC{} construct with index
$n$ have index $n$ with the only exception of the index of the auxiliary communications of the 
projection of the if and while constructs. 
In particular, for the projection of the if construct and for each role except the coordinator, 
we assign to the auxiliary input communications a fresh index $i$. As far as the coordinator is 
concerned instead, the auxiliary output communications in the if branch are indexed with 
$(i,\mathit{true})$ while the auxiliary communications in the else branch are indexed with 
$(i,\mathit{false})$, where $i$ is the fresh index associated to the target role.
The indexes of the auxiliary operations of the while construct projection are instead computed as follows:
\begin{itemize}
 \item for each non coordinator role we choose a pair of fresh indexes
   $i$ and $j$;
\item auxiliary inputs of non coordinator roles are annotated with the
  fresh index $i$;
 \item auxiliary outputs in non coordinator roles and the
   corresponding input in the coordinator are both annotated with the
   fresh index $j$;
 \item
 the first output auxiliary communications of the coordinator are
 indexed with $(i,\mathit{true})$, where $i$ is the fresh index
 corresponding to the target role;
 \item
 the last output auxiliary communications of the coordinator are
 indexed with $(i,\mathit{false})$, where $i$ is the fresh index
 corresponding to the target role.
\end{itemize}

As for \AIOC{}s, annotated \APOC{}s inherit the semantics of
\APOC{}s, since indexes are just decorations, with no effect on the behaviour.
There are however a few tricky points. In particular, we have to clarify how
indexes are managed when new constructs are introduced. In rule \ruleName{In}
the assignment inherits the index from the input primitive. In rule
\ruleName{While-unfold} the body is copied together with its indexes. In rule
\ruleName{Lead-\AdaptRule}, when applying the update ${\mathcal I}$, we
annotate ${\mathcal I}$ with indexes never used elsewhere and distinct, and
then generate the indexes for its projection as described above. Also, we
assign to the auxiliary communications introduced by rules
\ruleName{Lead-\AdaptRule} and \ruleName{Lead-\NoAdaptRule} indexes never used
elsewhere. Auxiliary communications introduced by rule \ruleName{\AdaptRule}
instead have to use the index of the corresponding communication introduced by
rule \ruleName{Lead-\AdaptRule} (the index can be passed by extending the
communication label).
We can now define \APOC{} events. As for \AIOC{} events, \APOC{} events
correspond to transitions of the \APOC{}.

\begin{definition}[\APOC{} events]\label{def:apocev}
An annotated \APOC{} network $\net$ contains the following events:
\begin{description}
\item[Communication events:] a sending event $\xi: \co{o^?}{r_2}$ in role
$r_1$ for each output $n: \cout{o^?}{e}{r_2}$  with global index $\xi$ in role
$r_1$; and a receiving event $\xi: \ci{o^?}{r_1}$ in role $r_2$ for each input
$n:\cinp{o^?}{x}{r_1}$ with global index $\xi$ in role $r_2$; we also denote
the sending event as $f_\xi$ or $[f_\xi]_{r_1}$; and the receiving event as
$t_\xi$ or $[t_\xi]_{r_2}$. Sending and receiving events correspond to the
transitions executing the communications.
\item[Assignment events:] an assignment event $\ev_\xi$ in role $r$ for each
assignment $n: x = e$ with global index $\xi$; the event
corresponds to the transition executing the assignment.
\item[Scope events:] a scope initialisation event $\uparrow_{\xi}$ and a scope
termination event $\downarrow_{\xi}$ for each $\pscope{n}{l}{r}{P}{\Delta}{S}$
or $\psscope{n}{l}{r}{P}$ with global index $\xi$. Scope events with the same
global index coincide, and thus the same event may belong to different roles;
the scope initialisation event corresponds to the transition performing or not
performing an update on the given scope for the role leading the update.
The scope termination event is just an auxiliary event (related to the
auxiliary interactions implementing the scope termination).
\item[If events:] a guard if-event  $\ev_\xi$ in role $r$ for each construct
$\ifthenKey{b}{P}{P'}{n}$ with global index $\xi$; the guard-if event
corresponds to the transition evaluating the guard of the if.
\item[While events:] a guard while-event  $\ev_\xi$ in role $r$ for each
construct $\whileKey{b}{P}{n}$ with global index $\xi$; the guard-while
event corresponds to the transition evaluating the guard of the while.
\end{description}
Let $\event(\net)$ denote the set of events of the network $\net$. A
sending and a receiving event with either the same global index $\xi$
or with global indexes differing only for replacing index $i$ with
$(i,\mathit{true})$ or $(i,\mathit{false})$ are called matching events.  We denote with
$\overline{\ev}$ an event matching event $\ev$. A communication event
is either a sending event or a receiving event. A communication event
is unmatched if there is no event matching it.
\end{definition} 
With a slight abuse of notation, we write $\event(P)$ to denote events
originated by constructs in process $P$, assuming the network $\net$
to be understood.

We used the same notations for events of the \AIOC{} and of the
\APOC{}. Indeed, the two kinds of events are strongly related (cfr.\
Lemma~\ref{lemma:ev}).

We can now define the causality relation among \APOC{} events.

\begin{definition}[\APOC{} causality relation]\label{def:causalapoc} Let us
consider an annotated \APOC{} network $\net$.  A causality relation $\leqapoc
\, \subseteq \, \event(\net) \times \event(\net)$ is a partial order among
events in $\net$. We define $\leqapoc$ as the minimum partial order
satisfying:

{\bf Sequentiality:} Let $P'\seqOp P''$ be a subterm of \APOC{} network
$\net$. If $\ev'$ is an event in $P'$ and $\ev''$ is an event in $P''$, both
in the same role $r$, then $\ev'
\leqapoc \ev''$.

{\bf Scope-coordinator:} Let $\pscope{n}{l}{r}{P}{\Delta}{S}$ be a subterm of
\APOC{} $\net$ in role $r$ with global index $\xi$. If $\ev'$ is an event in
$P$ then $\uparrow_{\xi} \leqapoc \ev' \leqapoc \downarrow_{\xi}$.

{\bf Scope-simple:} Let $\psscope{n}{l}{r}{P}$ be a subterm of
\APOC{} $\net$ in role $r'$ with global index $\xi$. If $\ev'$ is an event in
$P$ then $\uparrow_{\xi} \leqapoc \ev' \leqapoc \downarrow_{\xi}$.

{\bf Synchronisation:} For each pair of events $\ev$ and $\ev'$, $\ev \leq
\ev'$ implies $\overline{\ev} \leqapoc \ev'$.

{\bf If:} Let $\ifthenKey{b}{P}{P'}{n}$ be a subterm of \APOC{} network $\net$
 with global index $\xi$, let $\ev_\xi$ be the guard if-event in role $r$,
 then for every event $\ev$ in $P$ and for every event $\ev'$ in $P'$ we have
 $\ev_\xi \leqapoc \ev$ and $\ev_\xi \leqapoc \ev'$.

{\bf While:} Let $\whileKey{b}{P}{n}$ be a subterm of \APOC{} network $\net$
with global index $\xi$, let $\ev_\xi$ be the guard while-event in role $r$,
then for every event $\ev$ in $P$ we have $\ev_\xi \leqapoc \ev$.
\end{definition}

\begin{lemma}\label{lemma:ev} Given a \AIOC{} process $\mathcal{I}$, for each
state $\Sigma$ the \APOC{} network $\proj(\mathcal{I},\Sigma)$ is such that:
\begin{enumerate}
\item $\devent(\mathcal{I}) \subseteq \event(\proj(\mathcal{I},\Sigma))$;
\item $\forall \ev_1,\ev_2 \in \devent(\mathcal{I}). \ev_1 \leqaioc \ev_2
\Rightarrow  \ev_1 \leqapoc \ev_2 \vee \ev_1 \leqapoc \overline{\ev_2}$
\end{enumerate}
\end{lemma}
\begin{proof}
\begin{enumerate}
\item By definition of projection.
\item Let $\ev_1 \leqaioc \ev_2$. We have a case analysis on the condition
  used to derive the dependency.
\begin{description}
\item [Sequentiality:] Consider $\mathcal{I} = \mathcal{I}'\seqOp
\mathcal{I}''$. If  events are in the same role the implication follows from
the sequentiality of the $\leqapoc$.

Let us show that there exists an event $\ev''$ in an initial interaction of
$\mathcal{I}''$ such that either $\ev'' \leqapoc \ev_2$ or $\ev'' \leqapoc
\overline{\ev_2}$. The proof is by induction on the structure of
$\mathcal{I}''$. The only difficult case is sequential composition. Assume
$\mathcal{I}'' = \mathcal{I}_1; \mathcal{I}_2$. If $\ev_2 \in
\devent(\mathcal{I}_1)$ the thesis follows from inductive hypothesis. If
$\ev_2 \in \devent(\mathcal{I}_2)$ then by induction there exists an event
$\ev_3$ in an initial interaction of $\mathcal{I}_2$ such that $\ev_3 \leqapoc
\ev_2$ or $\ev_3 \leqapoc \overline{\ev_2}$. By synchronisation (Definition
\ref{def:causalapoc}) we have that $\overline{\ev_3} \leqapoc \ev_2$ or
$\overline{\ev_3}
\leqapoc
\overline{\ev_2}$. By connectedness for sequence we have that $\ev_3$ or
$\overline{\ev_3}$ are in the same role of an event $\ev_4$ in $\mathcal{I}'$.
By sequentiality (Definition \ref{def:causalapoc}) we have that $\ev_4
\leqapoc \ev_3$ or $\ev_4 \leqapoc
\overline{\ev_3}$. By synchronisation we have that $\overline{\ev_4} \leqapoc
\ev_3$ or $\overline{\ev_4} \leqapoc
\overline{\ev_3}$. The thesis follows from the inductive hypothesis on $\ev_4$
and by transitivity of $\leqapoc$.

Let us also show that there exists a final event $\ev''' \in
\devent(\mathcal{I}')$ such that $\ev_1 \leqapoc \ev'''$ or $\ev_1 \leqapoc
\overline{\ev'''}$. The proof is by induction on the structure of
$\mathcal{I}'$. The only difficult case is sequential composition. Assume
$\mathcal{I}' = \mathcal{I}_1; \mathcal{I}_2$. If $\ev_1 \in
\devent(\mathcal{I}_2)$ the thesis follows from inductive hypothesis. If
$\ev_1 \in \devent(\mathcal{I}_1)$ then the proof is similar to the one above,
finding a final event in $\mathcal{I}_1$ and applying sequentiality,
synchronisation, and transitivity.

The thesis follows from the two results above again by sequentiality,
synchronisation, and transitivity.
\item [Scope:] it means that either (1) $\ev_1 = \uparrow_{n}$  and $\ev_2$ is
an event in the scope or (2) $\ev_1 = \uparrow_{n}$  and  $\ev_2 =
\downarrow_{n}$, or (3) $\ev_1$ is an event in the scope and  $\ev_2 =
\downarrow_{n}$. We consider the first case since the third one is analogous
and the second one follows by transitivity. If $\ev_2$ is in the coordinator
then the thesis follows easily. Otherwise it follows thanks to the auxiliary
synchronisations with a reasoning similar to the one for sequentiality.
\item [Synchronisation: ] it means that $\ev_1$ is a sending event and $\ev_2$
is the corresponding receiving event, namely $\ev_1 = \overline{\ev_2}$ .
Thus, since $\ev_2 \leqapoc \ev_2$ then $\overline{\ev_2} \leqapoc \ev_2$.
\item [If:] it means that $\ev_1$ is the evaluation of the guard and $\ev_2$
is in one of the two branches. Thus, if $\ev_2$ is in the coordinator then the
thesis follows easily. Otherwise it follows thanks to the auxiliary
synchronisations with a reasoning similar to the one for sequentiality.
\item [While:] it means that $\ev_1$ is the evaluation of the guard and
$\ev_2$ is in the body of the while. Thus, if $\ev_2$ is in the coordinator
then the thesis follows easily. Otherwise it follows thanks to the auxiliary
synchronisations with a reasoning similar to the one for sequentiality.
\end{description}
%
%
%
\end{enumerate}
\end{proof}

We can now define a notion of conflict between (\AIOC{} and \APOC{}) events,
relating events which are in different branches of the same conditional.

\begin{definition}[Conflicting events]\label{def:conflict} Given a \AIOC{}
process $\mathcal{I}$ we say that two events $\ev, \ev' \in
\devent(\mathcal{I})$ are conflicting if they belong to different branches of
the same if construct, i.e. there exists a subprocess
$\ifthen{b}{\mathcal{I}'}{\mathcal{I}''}$ of $\mathcal{I}$ such that $\ev
\in \devent(\mathcal{I}') \wedge \ev'
\in \devent(\mathcal{I}'')$ or $\ev'
\in \devent(\mathcal{I}') \wedge \ev
\in \devent(\mathcal{I}'')$.

Similarly, given a \APOC{} network $\net$, we say that two events $\ev, \ev'
\in \event(\net)$ are conflicting if they belong to different branches of the
same if construct, i.e. there exists a subprocess $\ifthen{b}{P}{P'}$ of
$\net$ such that $\ev
\in \event(P) \wedge \ev'
\in \event(P')$ or $\ev'
\in \event(P) \wedge \ev
\in \event(P')$.
\end{definition}

\APOC{}s resulting from the projection of well-annotated connected \AIOC{}s
enjoy useful properties.

\begin{definition}[Well-annotated \APOC{}]
\label{defin:synchwa} An annotated \APOC{} network $\net$ is well-anno\-tated
for its causality relation $\leqapoc$ if the following conditions hold:
\begin{description}
\item[{\sc C1}] for each global index $\xi$ there are at most two
communication events with global index $\xi$ and, in this case, they are
matching events;
\item[{\sc C2}] only events which are minimal according to $\leqapoc$ may
correspond to enabled transitions;
\item[{\sc C3}] for each pair of non-conflicting sending events $[f_\xi]_r$
and $[f_{\xi'}]_r$ on the same operation $o^?$ with the same target $s$ such
that $\xi \neq \xi'$ we have $[f_\xi]_r \leqapoc [f_{\xi'}]_r$ or
$[f_{\xi'}]_r \leqapoc [f_\xi]_r$;
\item[{\sc C4}] for each pair of non-conflicting receiving events $[t_\xi]_s$
and $[t_{\xi'}]_s$ on the same operation $o^?$ with the same sender $r$ such
that $\xi \neq \xi'$ we have $[t_\xi]_s \leq [t_{\xi'}]_s$ or $[t_{\xi'}]_s
\leq [t_\xi]_s$;
\item[{\sc C5}] if $\ev$ is an event inside a scope with global index $\xi$
then its matching event $\overline{\ev}$ (if it exists) is inside a scope with
the same global index.
\item[{\sc C6}] if two events have the same index but different global indexes
then one of them is inside a while with global index $\xi_1$, let us call it
$\ev_1$, and the other, $\ev_2$, is not. Furthermore, $\ev_2 \leqapoc
\ev_{\xi_1}$ where $\ev_{\xi_1}$ is the guarding while-event of the while with
global index $\xi_1$.
\end{description}
\end{definition}
%


%


Update, conditional choice, and iteration at the \AIOC{} level happen in
one step, while they correspond to many steps of the projected
\APOC{}. Also, scope execution introduces auxiliary communications which have
no correspondence in the \AIOC{}.  Thus, we define the function $\upd$ that
bridges this gap. More precisely, function $\upd$ is obtained as the
composition of two functions, a function $\prop$ that completes the execution
of \AIOC{} actions which have already started, and a function $\ssim$ that
eliminates all the auxiliary closing communications.

\begin{definition}[$\upd$ function]\label{def:upd} Let $\net$ be an annotated
\APOC{} (we drop annotations if not relevant). The $\upd$ function is defined
as the composition of a function $\prop$ and a function $\ssim$. Thus,
$\upd(\net) = \ssim(\prop(\net))$. Network $\prop(\net)$ is obtained from
$\net$ by repeating the following operations while possible:
\begin{enumerate}
\item for each $\cout{o^*_n}{\mathit{true}}{r'}$ enabled, replace every
$\cinp{o^*_n}{x_n}{r}\seqOp  \while{x_n}{P\seqOp
\cout{o^*_n}{\texttt{ok}}{r}\seqOp \cinp{o^*_n}{x_n}{r}}$
not inside another while construct, with $P\seqOp
\cout{o^*_n}{\texttt{ok}}{r}\seqOp \cinp{o^*_n}{x_n}{r}\seqOp
\while{x_n}{P\seqOp  \cout{o^*_n}{\texttt{ok}}{r}\seqOp
\cinp{o^*_n}{x_n}{r}}$;
and replace $\cout{o^*_n}{\mathit{true}}{r'}$ with $\one$.
\item for each $\cout{o^*_n}{\mathit{false}}{r'}$ enabled, replace every
$\cinp{o^*_n}{x_n}{r}\seqOp  \while{x_n}{P\seqOp
\cout{o^*_n}{\texttt{ok}}{r}\seqOp \cinp{o^*_n}{x_n}{r}}$ not inside another
while construct, with $\one$; and replace $\cout{o^*_n}{\mathit{false}}{r'}$
with $\one$.
\item  for each $\while{x_n}{P\seqOp \cout{o^*_n}{\texttt{ok}}{r}\seqOp
\cinp{o^*_n}{x_n}{r}}$ enabled not inside another while construct, such that
$x_n$ evaluates to \texttt{true} in the local state, replace it with $P\seqOp
\cout{o^*_n}{\texttt{ok}}{r}\seqOp \cinp{o^*_n}{x_n}{r}\seqOp
\while{x_n}{P\seqOp  \cout{o^*_n}{\texttt{ok}}{r}\seqOp
\cinp{o^*_n}{x_n}{r}}$.
\item  for each $\while{x_n}{P\seqOp \cout{o^*_n}{\texttt{ok}}{r}\seqOp
\cinp{o^*_n}{x_n}{r}}$ enabled not inside another while construct, such that
$x_n$ evaluates to \texttt{false} in the local state,  replace it with $\one$.
\item for each $\cout{o^*_n}{\mathit{true}}{r'}$ enabled, replace every
$\cinp{o^*_n}{x_n}{r}\seqOp  \ifthen{x_n}{P'}{P''}$ not inside a while
construct, with $P'$; and replace $\cout{o^*_n}{\mathit{true}}{r'}$ with
$\one$.
\item for each $\cout{o^*_n}{\mathit{false}}{r'}$ enabled, replace every
$\cinp{o^*_n}{x_n}{r}\seqOp  \ifthen{x_n}{P'}{P''}$ not inside a while
construct, with $P''$; and replace  $\cout{o^*_n}{\mathit{false}}{r'}$ with
$\one$.
\item for each $\ifthen{x_n}{P'}{P''}$ enabled such that $x_n$ evaluates to
\texttt{true} in the local state, replace it with $P'$.
\item for each $\ifthen{x_n}{P'}{P''}$ enabled such that $x_n$ evaluates to
\texttt{false} in the local state, replace it with $P''$.
\item for each $\cout{o^*_n}{P}{s}$ enabled, replace every
$\psscope{n}{l}{r}{P'}$ in role $s$ not inside a while construct, with $P$,
and replace $\cout{o^*_{n}}{P}{s}$ with $\one$.
%
%
\item for each $\cout{o^*_{n}}{\texttt{no}}{s}$ enabled, replace every
$\psscope{n}{l}{r}{P'}$ in the role $s$ not inside a while construct, with
$P'$  and replace $\cout{o^*_{n}}{P}{s}$ with $\one$.
\end{enumerate}
%
Network $\ssim(\net)$ is obtained from $\net$ by repeating the following
operations while possible:
\begin{itemize}%
\item replace each $\cout{o^*_n}{\texttt{ok}}{r}$,
$\cout{o^*_{n}}{\texttt{ok}}{r}$, $\cinp{o^*_n}{\_}{r}$ or
$\cinp{o^*_{n}}{\_}{r}$ not inside a while construct with $\one$.
\item replace each operation occurrence of the form $\freshops{n}{o^?}$ with
$o^?$.
\end{itemize} Furthermore $\ssim$ may apply 0 or more times the following
operation:
\begin{itemize}
\item \label{pruning_def:one}replace a subterm $\one \seqOp P$ by $P$ or a
subterm $\one \mid P$ by $P$.
\end{itemize}
\end{definition}




The result below proves that in a well-annotated \APOC{} only transitions
corresponding to events minimal w.r.t.\ the causality relation $\leqapoc$ may
be enabled.

\begin{lemma}\label{lemma:minimalevent} If $\net$ is a \APOC{}, $\leqapoc$
its causality relation and $\ev$ is an event corresponding to a transition
enabled in $\net$ then $\ev$ is minimal w.r.t.\ $\leqapoc$.
\end{lemma}
\begin{proof} The proof is by contradiction. Suppose $\ev$ is enabled but not
minimal, i.e.\ there exists $\ev'$ such that $\ev' \leqapoc \ev$. If there is
more than one such $\ev'$ consider the one such that the length of the
derivation of $\ev' \leqapoc \ev$ is minimal. This derivation should have
length one, and following  Definition~\ref{def:causalapoc} it may result from
one of the following cases:
\begin{itemize}
\item Sequentiality: $\ev' \leqapoc \ev$ means that $\ev' \in \event(P')$,
$\ev \in
\event(P'')$, and $P' \seqOp P''$ is a subterm of $\net$. Because of the
semantics of sequential composition $\ev$ cannot be enabled.
\item Scope: let $\pscope{n}{l}{r}{P}{\Delta}{S}$ or $\psscope{n}{l}{r}{P}$ be
a subprocess of $\net$ with global index $\xi$. We have the following cases:
\begin{itemize}
\item $\ev' = \uparrow_{\xi}$ and $\ev \in \event(P)$, and this implies that
$\ev$ cannot be enabled since if $\ev'$ is enabled then the rules
\ruleName{\AdaptRule} or \ruleName{\NoAdaptRule} for the evolution of the
scope have not been applied yet;
\item $\ev' = \uparrow_{\xi}$ and $\ev = \downarrow_{\xi}$: this is trivial,
since  $\downarrow_{\xi}$ is an auxiliary event and no transition corresponds
to it;
\item $\ev' \in \event(P)$ and $\ev = \downarrow_{\xi}$, but this is
impossible since if $\ev'$ is enabled there is no event $\ev$ because the
events $\uparrow_{\xi}$ and $\downarrow_{\xi}$ disappear as soon as the rule
\ruleName{Lead-\AdaptRule} or \ruleName{Lead-\NoAdaptRule} is performed.
\end{itemize}
\item If: $\ev \leqapoc \ev'$ means that $\ev$ is the evaluation of the guard
of the subterm $\ifthenKey{x_n}{P'}{P''}{n}$ and $\ev' \in \event(P') \cup
\event(P'')$. Event $\ev'$ cannot be enabled because of the semantics of if.
\item While: $\ev \leqapoc \ev'$ means that $\ev$ is the evaluation of the
guard of the subterm $\whileKey{x_n}{P}{n}$ and $\ev' \in \event(P)$. Event
$\ev'$ cannot be enabled because of the semantics of while.
\end{itemize}
\end{proof}
%
%
The following result shows that if an interaction is performed then the two
executed events are matching events.
\begin{lemma}\label{lemma:executematching} If $\net$ is a well-annotated
\APOC{} and $\tuple{\ambientN,\net} \arro{\commLabel{o^?}{r_1}{v}{r_2}{x}}
\tuple{\ambientN,\net'}$ then the two executed events are matching events.
\end{lemma}
\begin{proof} By definition of \APOC{} semantics we have that the transition
\\ $\tuple{\ambientN,\net}
\arro{\commLabel{o^?}{r_1}{v}{r_2}{x}} \tuple{\ambientN,\net'}$ can be
generated only by the \ruleName{synch} rule.
Then, we have that the two events are on the same operation and that
$r_2$ is the target of the first event. Assume that they are not
matching events. Then for the definition of well-annotated \APOC{},
they are either conflicting or in the causality relation. In the first
case, none of them can be enabled by Definition \ref{def:causalapoc}
since they are inside an \emph{if} construct.  In the second case
thanks to Lemma~\ref{lemma:minimalevent}, at least one of them cannot
be enabled since it is not minimal. This is a contradiction, thus they
are matching events.
\end{proof}

We now prove that all the \APOC{}s obtained as projection of
well-annotated connected \AIOC{}s are well-annotated.
%
\begin{lemma}\label{lemma:IOCwell} Let ${\mathcal I}$ be a well-annotated
connected \AIOC{} process, and $\Sigma$ a state. Then its projection $\net =
\proj({\mathcal I},\Sigma)$ is a well-annotated \APOC{} network w.r.t.
$\leqapoc$.
\end{lemma}
\begin{proof} We have to prove that $\proj(\mathcal{I},\Sigma)$ satisfies the
conditions of Definition~\ref{defin:synchwa} of well-annotated \APOC{}:
\begin{description}
\item[{\sc C1}] For each global index $\xi$ there are at most two
  communication events with global index $\xi$ and, in this case, they
  are matching events. The condition follows by the definition of the
  projection function, observing that in well-annotated \AIOC{}s, each
  construct has its own index, and different indexes are mapped to
  different global indexes. Note that the two auxiliary input
  communications in the projection of a while construct on a non
  coordinating role have the same index but different global indexes.
\item[{\sc C2}] Only events which are minimal according to $\leqapoc$ may
correspond to enabled transitions. This condition follows from
Lemma~\ref{lemma:minimalevent}.
\item[{\sc C3}] For each pair of non-conflicting sending events $[f_\xi]_{r}$
and $[f_{\xi'}]_{r}$ on the same operation $o^?$ and with the same target such
that $\xi \neq \xi'$ we have $[f_\xi]_r \leqapoc [f_{\xi'}]_r$ or
$[f_{\xi'}]_r \leqapoc [f_{\xi}]_r$. Note that the two events are in the same
role, thus w.l.o.g. we can assume that there exist two processes $P,P'$ such
that $[f_\xi]_r \in \event(P)$ and $[f_{\xi'}]_r \in \event(P')$ and that one
among $P\seqOp P'$, $P \parOpP P'$, and $\ifthen{b}{P}{P'}$ is a subprocess of
$\net$.

Since ${\mathcal I}$ is connected for parallel, by Definition
\ref{def:connectedness} and by definition of the projection function the
second case can never happen. Similarly, since the events are non-conflicting
by Definition \ref{def:conflict} the third case can never happen. If $P\seqOp
P'$ is a subprocess of $\net$ then by sequentiality (Definition
\ref{def:causalapoc}) we have the thesis.

\item[{\sc C4}] Similar to the previous case.
\item[{\sc C5}] By definition of the projection function.
\item[{\sc C6}] By definition of well-annotated \AIOC{} and of
  projection the only case where there are two events with the same
  index and different global indexes is for the auxiliary
  communications in the projection of the while construct, where the
  conditions hold by construction.
\end{description}
\end{proof}

The next lemma shows that for every set of updates
$\rules$ the \APOC{} $\net$ and $\upd(\net)$ have the same set of weak traces.

\begin{lemma}\label{lemma:uptoupd} Let $\net$ be a \APOC{}. The following
properties hold:
\begin{enumerate}
\item if $\tuple{\ambientN,\upd(\net)} \arro{\eta} \tuple{\ambientN,\net'}$
with $\eta \in \{
\commLabel{o^?}{r_1}{v}{r_2}{x},
\tick, {\mathcal I},\texttt{\NoAdaptLabel}, \tau \}$ then
there exist $\net''$ s.t.~$\tuple{\ambientN,\net}
\arro{\eta_1} \dots \arro{\eta_k} \arro{\eta} \tuple{\ambientN,\net''}$ where
$\eta_i \in
\{\commLabel{o^*}{r_1}{v}{r_2}{x},\tau\}$ and $\upd(\net'') = \upd(\net')$.
\item if $\tuple{\ambientN,\net} \arro{\eta} \tuple{\ambientN,\net'}$ for
$\eta \in
\{\commLabel{o^?}{r_1}{v}{r_2}{x}, \tick,
{\mathcal I},\texttt{\NoAdaptLabel}, \tau\}$, then one of
the following holds: (A) $\tuple{\ambientN,\upd(\net)} \arro{\eta}
\tuple{\ambientN,\net''}$ such that $\upd(\net') = \upd(\net'')$, or (B)
$\upd(\net) = \upd(\net')$ and $\eta \in
\{\commLabel{o^*}{r_1}{v}{r_2}{x},\tau\}$;
%
\end{enumerate}
\end{lemma}
\begin{proof}
\mbox{}\\
\begin{enumerate}
\item The $\upd$ function corresponds to perform weak transitions, namely
transitions with labels in $\{\commLabel{o^*}{r_1}{v}{r_2}{x},\tau\}$. $\net$
may perform the enabled weak transitions that correspond to the application of
$\upd$ reducing to $\net'''$. Then, $\eta$ is enabled also in $\net'''$ and we
have $\tuple{\ambientN,\net'''} \arro{\eta} \tuple{\ambientN,\net''}$. At this
point we have that $\net''$ and $\net'$ may differ only for communication
primitives corresponding to weak transitions, removed by $\upd$.
\item Either the transition with label $\eta$ corresponds to one of the
transitions executed by function $\upd$ or not. In the first case statement
(B) holds trivially. Otherwise transition labeled by $\eta$ is still enabled
in $\upd(\net)$ and the thesis follows.
\end{enumerate}
\end{proof}


We now prove a few properties of transitions with label $\tick$.

\begin{lemma}\label{lemma:tick} If $\tuple{\state,\rules,{\mathcal I}}$
has a transition with label $\tick$ then, for each role $s \in
\roles({\mathcal I})$, $\role{\pi({\mathcal I},s),\state_s}{s}$ has a
transition with label $\tick$ and vice versa.
\end{lemma}
\begin{proof} By structural induction on ${\mathcal I}$.
\end{proof}

%
The next lemma shows that if two matching events are enabled in the projection
of a \AIOC{}, then the corresponding interaction is enabled in the \AIOC{}.
\begin{lemma}\label{lemma:synchenabled} Let ${\mathcal I}$ be a \AIOC{}
obtained from a well-annotated connected \AIOC{} via $0$ or more transitions
and $n:\comm{o^?}{r_1}{e}{r_2}{x}$ be an interaction in ${\mathcal I}$. If $n:
\cout{o^?}{e}{r_1} $ and $n:\cinp{o^?}{x}{r_2}$ are matching events and are
both enabled in $\proj({\mathcal I},\Sigma)$ then $n:
\comm{o^?}{r_1}{e}{r_2}{x} $ is enabled.
\end{lemma}
\begin{proof} 
Note that ${\mathcal I}$ is well-annotated and connected for parallel.

If ${\mathcal I}$ is also connected for sequence, then the
proof is by structural induction on ${\mathcal I}$.  The cases for $\one$,
$\zero$, and scopes, if, and while constructs are trivial.
For parallel composition just consider that since the two events have the same
global index then they are from the same component, and the thesis follows by
inductive hypothesis.
Let us consider sequential composition. Suppose ${\mathcal I} = {\mathcal
I}'\seqOp{\mathcal I}''$. If $n: \comm{o^?}{r_1}{e}{r_2}{x} \in \mathcal{I}'$
then the thesis follows by inductive hypothesis. Otherwise, by inductive
hypothesis $n: \comm{o^?}{r_1}{e}{r_2}{x}$ is enabled in ${\mathcal I}''$.
Thus, $r_1 \rightarrow r_2 \in \transI({\mathcal I}'')$. From connectedness
for sequence $\forall s_1 \rightarrow s_2 \in \transF(\mathcal{I}')$ then
$\{r_1,r_2\} \cap \{s_1,s_2\} \neq \emptyset$. This is not possible since
otherwise at least one of the events $n: \cout{o^?}{e}{r_1} $ and
$n:\cinp{o^?}{x}{r_2}$ would not be enabled. Thus, the only possibility is
$\transF(\mathcal{I}') = \emptyset$. This implies that  ${\mathcal I}'$ has a
transition with label $\tick$. Thus, $n: \comm{o^?}{r_1}{e}{r_2}{x}$ is
enabled in $\mathcal{I}$.

If ${\mathcal I}$ is not connected for sequence, then in the projected
\APOC{} some more transitions may be enabled, but no required
transitions may be disabled, thus the thesis follows.
%
\end{proof}


\begin{definition}[Weak System Bisimilarity]\label{def:bisim} A  weak system
bisimulation is a relation $R$ between \AIOC{} systems and \APOC{} systems
such that if \\
$(\tuple{\state,\rules,{\mathcal
I}},\tuple{\rules',\net}) \in R$ then:
\begin{itemize}
\item if $\tuple{\state,\rules,{\mathcal I}} \arro{\mu}
\tuple{\state'',\rules'',{\mathcal I}''}$ then\\
$\tuple{\rules',\net} \arro{\eta_1} \dots \arro{\eta_k} \arro{\eta}
\tuple{\rules''',\net'''}$ with $\forall i \in [1..k], \eta_i \in
\{\commLabel{o^*}{r_1}{v}{r_2}{x},\tau\}$ and
$(\tuple{\state'',\rules'',{\mathcal
I}''},\tuple{\rules''',\net'''}) \in R$ and $\eta = \mu$ or $\eta =
\commLabel{\freshops{n}{o^?}}{r_1}{v}{r_2}{x}$ and $\mu =
\commLabel{o^?}{r_1}{v}{r_2}{x}$;\\

\item if $\tuple{\rules',\net} \arro{\eta}
\tuple{\rules''',\net'''}$ with $\eta \in \{
\commLabel{o^?}{r_1}{v}{r_2}{x};
\tick\seqOp {\mathcal I} \seqOp \texttt{\NoAdaptLabel}
\seqOp \rules''', \tau \}$ then one of the following two holds:\\
\begin{itemize}
\item $\tuple{\state,\rules,{\mathcal I}} \arro{\mu}
\tuple{\state'',\rules,{\mathcal I}''}$ , with $\eta = \mu$ or $\eta =
\commLabel{\freshops{n}{o^?}}{r_1}{v}{r_2}{x}$ and $\mu =
\commLabel{o^?}{r_1}{v}{r_2}{x}$ and it holds that
$(\tuple{\state'',\rules'',{\mathcal I}''},$
$\tuple{\rules''',\net'''}) \in R$;
\item $\eta \in
\{\commLabel{o^*}{r_1}{v}{r_2}{x},\commLabel{o^*}{r_1}{X}{r_2}{\_},\tau\}$ and
it holds that $(\tuple{\state,\rules,{\mathcal I}},$
$\tuple{\rules''',\net'')} \in R$\\
\end{itemize}
\end{itemize}

Weak system bisimilarity $\bisim$ is the largest weak system bisimulation.
\end{definition}

The following result states that weak system bisimilarity implies  weak trace
equivalence.

\begin{lemma}\label{lemma:bis2trsynch} Let
$\tuple{\state,\rules,{\mathcal I}}$ be a \AIOC{} system and
$\tuple{\rules',\net}$ a \APOC{} system. \\If
$\tuple{\state,\rules,{\mathcal I}} \bisim \tuple{\rules',\net}$
then the \AIOC{} system $\tuple{\state,\rules,{\mathcal I}}$ and the
\APOC{} system $\tuple{\rules',\net}$ are weak trace equivalent.
\end{lemma}
\begin{proof} The proof is by coinduction.  Take a \AIOC{} trace
$\mu_1,\mu_2,\dots$ of the
\AIOC{} system. From bisimilarity, the \APOC{} system has a transition with
label $\eta_1$ matching $\mu_1$. After the transition, the \AIOC{} system and
the \APOC{} system are again bisimilar. Thus the
\APOC{} system has a trace $\eta_2,\dots$ matching $\mu_2,\dots$. By
composition the \APOC{} system has a trace $\eta_1,\eta_2,\dots$ as desired.
The opposite direction is analogous.
\end{proof}

We can now prove our main theorem, that states that given a connected
well-annota\-ted \AIOC{} process ${\mathcal I}$ and a state $\Sigma$
the \APOC{} network obtained as its projection has the same behaviours
of ${\mathcal I}$.

\begin{theorem} For each initial, connected \AIOC{} process ${\mathcal I}$,
each state $\state$, and each set of updates $\rules$,
the \AIOC{} system $\tuple{\state, \rules,{\mathcal I}}$ and the \APOC{}
system $\tuple{\rules,\proj({\mathcal I},\state)}$ are weak trace
equivalent.
\end{theorem}

\begin{proof} We prove that the relation $R$ below is a weak system
bisimulation. {\footnotesize
\[ R=\sset{(\tuple{\state,\rules,{\mathcal I}},\tuple{\rules,\net})}{ 
\upd(\net) = \proj({\mathcal I},\Sigma),\\
\devent(\mathcal{I}) \subseteq \event(\prop(\net)),\\
	\forall \ev_1, \ev_2 \in \devent(\mathcal{I}) . \\
	\ev_1 \leqaioc \ev_2 \Rightarrow\\ \ev_1 \leqapoc \ev_2 \vee \ev_1
\leqapoc \overline{\ev_2} }
\] } where ${\mathcal I}$ is obtained from a well-annotated connected \AIOC{}
via $0$ or more transitions and $\upd(\net)$ is a well-annotated \APOC{}.

To ensure that proving that the relation above is a bisimulation implies our
thesis, let us show that the pair $(\tuple{\state,\rules,{\mathcal
I}},\tuple{\rules,\proj(\mathcal{I},\Sigma)})$ from the theorem statement
belongs to $R$. Note that here $\mathcal{I}$ is well-annotated and connected,
and for each such $\mathcal{I}$ we have  $\upd(\proj(\mathcal{I},\Sigma)) =
\proj(\mathcal{I},\Sigma)$. From Lemma~\ref{lemma:IOCwell} $\proj({\mathcal
I},\Sigma)$ is well annotated, thus $\upd(\proj(\mathcal{I},\Sigma))$ is well
annotated.
      
Observe that $\prop$ is the identity on $\proj({\mathcal I},\Sigma)$, thus
from Lemma~\ref{lemma:ev} we have that the conditions $\devent(\mathcal{I})
\subseteq \event(\prop(\net))$ and    $\forall \ev_1, \ev_2 \in
\devent(\mathcal{I}) .
    \ev_1 \leqaioc \ev_2 \Rightarrow \ev_1 \leqapoc \ev_2 \vee \ev_1 \leqapoc
\overline{\ev_2}$ are satisfied. We now prove that $R$ is a weak system
bisimulation. From Lemma~\ref{lemma:bis2trsynch}, this implies weak trace
equivalence.

%


To prove that $R$ is a weak system bisimulation it is enough to prove that for
each $(\tuple{\state,\rules,{\mathcal I}},\tuple{\rules,\net})$
where $\net=\proj({\mathcal I,\Sigma})$ we have:
\begin{itemize}
\item if $\tuple{\state,\rules,{\mathcal I}} \arro{\mu}
\tuple{\state'',\rules'',{\mathcal I}''}$ then\\
$\tuple{\rules,\net} \arro{\eta} \tuple{\rules''',\net'''}$\\
with $(\tuple{\state'',\rules'',{\mathcal
I}''},\tuple{\rules''',\net'''}) \in R$ and\\ $\eta = \mu$ or $\eta =
\commLabel{\freshops{n}{o^?}}{r_1}{v}{r_2}{x}$ and $\mu =
\commLabel{o^?}{r_1}{v}{r_2}{x}$\seqOp
\item if $\tuple{\rules,\net} \arro{\eta}
\tuple{\rules''',\net'''}$ with\\ $\eta \in \{
\commLabel{o^?}{r_1}{v}{r_2}{x}\seqOp
\tick\seqOp {\mathcal I} \seqOp \texttt{\NoAdaptLabel}
\seqOp$ $,\rules'''\seqOp \tau \}$ then\\
$\tuple{\state,\rules,{\mathcal I}} \arro{\mu}
\tuple{\state'',\rules,{\mathcal I}''}$ and \\
$(\tuple{\state'',\rules'',{\mathcal
I}''},\tuple{\rules''',\net'''}) \in R$ and $\eta = \mu$ or $\eta =
\commLabel{\freshops{n}{o^?}}{r_1}{v}{r_2}{x}$ and $\mu =
\commLabel{o^?}{r_1}{v}{r_2}{x}$.
\end{itemize} In fact, consider $\net$ with $\upd(\net)=\proj({\mathcal
I},\Sigma)$. The case for labels $\Sigma, \rules$ is trivial. If
$\tuple{\state,\rules,{\mathcal I}} \arro{\mu}
\tuple{\state'',\rules'',{\mathcal I}''}$, then by hypothesis
$\upd(\net)
\arro{\eta} \net'''$. The thesis follows from Lemma~\ref{lemma:uptoupd} (case
one).
If instead $\tuple{\rules,\net} \arro{\eta}
\tuple{\rules''',\net'''}$ with $\eta \in \{
\commLabel{o^?}{r_1}{v}{r_2}{x},
\tick, {\mathcal I},\texttt{\NoAdaptLabel}, \tau \}$ then
thanks to Lemma~\ref{lemma:uptoupd} we have one of the following: (A)
$\upd(\net) \arro{\eta} \net''$ such that $\upd(\net''') = \upd(\net'')$, or
(B) $\upd(\net) = \upd(\net''')$ and $\eta \in
\{\commLabel{o^*}{r_1}{v}{r_2}{x},\commLabel{o^*}{r_1}{X}{r_2}{\_},\tau\}$. In
case (A) we have $\tuple{\rules,\upd(\net)} \arro{\eta}
\tuple{\rules'',\net''}$. Then we have
$\tuple{\state,\rules,{\mathcal I}} \arro{\mu}
\tuple{\state'',\rules'',{\mathcal I}''}$ and
$(\tuple{\state'',\rules'',{\mathcal I}''}, $
$\tuple{\rules'',\net''}) \in R$. The thesis follows since
$\upd(\net''')=\upd(\net'')$. In case (B) the step is matched by the
\AIOC{} by staying idle, following the second option in the definition of weak
system bisimilarity.\\

Thus, we have to prove the two conditions above. The proof is by structural
induction on the \AIOC{} ${\mathcal I}$. All the subterms of a well-annotated
connected \AIOC{} are well-annotated and connected, thus the induction can be
performed. We consider both challenges from the \AIOC{} ($\rightarrow$) and
from the \APOC{} ($\leftarrow$). The case for label $\tick$ follows from
Lemma~\ref{lemma:tick}. The case for labels $\Sigma, \rules$ is trivial. Let
us consider the other labels, namely $\commLabel{o^?}{r_1}{v}{r_2}{x},
{\mathcal I},\texttt{\NoAdaptLabel}$, and $\tau$.

Note that no transition (at the \AIOC{} or at the \APOC{} level) with one of
these labels can change the set of updates $\rules$.
Thus, in the following, we will not write it. Essentially, we
will use \AIOC{} processes and \APOC{} networks instead of
\AIOC{} systems and \APOC{} systems respectively. Note that \APOC{} networks
also include the state, while this is not the case for \AIOC{} processes. For
\AIOC{} processes, we assume to associate to them the state $\state$, and
comment on its changes whenever needed.

\begin{description}
\item[Case $\one$, $\zero$:] trivial.
\item[Case $n:\assign{x}{r}{e}$: ] the assignment changes the global state in
  the \AIOC{}, and the local state of role $r$ in the \APOC{} in a
  corresponding way.
\item[Case $n:\comm{o^?}{r_1}{e}{r_2}{x}$:] trivial unless the interaction has
  been created by an update step. In this last case, note that the
  mismatch on the name of the operation, namely between $\freshops{n}{o^?}$ in
  the \APOC{} and $o^?$ in the \AIOC{}, is solved thanks to the definition of
  weak system bisimilarity.
\item[Case ${\mathcal I}\seqOp{\mathcal I}'$: ] from the definition of the
projection function we have that \\ 
$\net = \parallel_{r \in
\roles(\mathcal{I}\seqOp\mathcal{I}')} (\pi({\mathcal I},r)\seqOp
\pi({\mathcal I}',r),\Sigma_r)_r$.

\begin{description}
\item[$\rightarrow$] Assume that ${\mathcal I}\seqOp{\mathcal I}' \arro{\mu}
{\mathcal I}''$ with $\mu \in
\{\commLabel{o^?}{r_1}{v}{r_2}{x}\seqOp \mathcal{I}\seqOp
\texttt{\NoAdaptLabel},\tau \}$.  There are two possibilities: either
${\mathcal I}
\arro{\mu} {\mathcal I}'''$ and ${\mathcal I}'' = {\mathcal
I}'''\seqOp{\mathcal I}'$ or ${\mathcal I}$ has a transition with label
$\tick$ and ${\mathcal I}'
\arro{\mu} {\mathcal I}''$. In the first case by inductive hypothesis
$\parallel_{r \in \roles(\mathcal{I})} \role{\pi({\mathcal I},r),\Sigma_r}{r}
\arro{\eta} \net'''$ with $\eta$ corresponding to $\mu$ and $\upd(\net''') =
\parallel_{r \in
\roles(\mathcal{I})}
\role{\pi({\mathcal I}''',r),\Sigma_r'}{r}$. Thus $\parallel_{r \in
\roles(\mathcal{I})}
\role{\pi({\mathcal I},r)\seqOp\pi({\mathcal I}',r),\Sigma_r}{r} \arro{\eta}
\net$ and we have \\ $\upd(\net)=\parallel_{r \in \roles(\mathcal{I})}
\role{\pi({\mathcal I}''',r)\seqOp\pi({\mathcal I}',r),\Sigma_r'}{r}$. If
$\roles(\mathcal{I}') \subseteq \roles(\mathcal{I})$ then the thesis follows.
Otherwise roles in  $\roles(\mathcal{I}') \setminus \roles(\mathcal{I})$ are
unchanged. Note however that the projection of $\mathcal{I}$ on these roles is
a term composed only by $\one$s, which can be removed by function $\upd$.

If ${\mathcal I}$ has a transition with label $\tick$ and ${\mathcal I}'
\arro{\mu} {\mathcal I}''$ then by inductive hypothesis $\proj({\mathcal
I}',\Sigma)
\arro{\eta} \net''$ with $\eta$ corresponding to $\mu$ and
$\upd(\net'')=\proj({\mathcal I}'',\Sigma')$. The thesis follows since, thanks
to Lemma~\ref{lemma:tick}, $\proj({\mathcal I}\seqOp{\mathcal I}',\Sigma)
\arro{\eta} \net$ and $\upd(\net) = \proj({\mathcal I}'',\Sigma')$.

Note that, in both the cases, conditions on events follow by inductive
hypothesis.

\item[$\leftarrow$] Assume that
\begin{multline*}
\net =
\parallel_{r \in \roles(\mathcal{I}\seqOp\mathcal{I}')} \role{\pi({\mathcal
I},r)\seqOp\pi({\mathcal I}',r),\Sigma_r}{r} \arro{\eta}
 \parallel_{r \in \roles(\mathcal{I}\seqOp\mathcal{I}')}
 \role{P_r,\Sigma_r'}{r}
\end{multline*} with $\eta \in  \{\commLabel{o^?}{r_1}{v}{r_2}{x},
{\mathcal I},\texttt{\NoAdaptLabel}, \tau\}$. We have a
case analysis on $\eta$.

If $\eta=\commLabel{o^?}{r_1}{v}{r_2}{x}$ then $\role{\pi({\mathcal
I}\seqOp{\mathcal I'},r_1),
\Sigma_{r_1}}{r_1}
\arro{\coutLabel{o^?}{v}{r_2}:r_1} \role{P_{r_1}, \Sigma_{r_1}}{r_1}$ and also
$\role{\pi({\mathcal I}\seqOp{\mathcal I'},r_2),\Sigma_{r_2}}{r_2}
\arro{\cinpLabel{o^?}{x}{v}{r_1}:r_2} \role{P_{r_2},\Sigma_{r_2}}{r_2}$.
The two events should have the same global index thanks to
Lemma~\ref{lemma:executematching}. Thus, they are either both from ${\mathcal
I}$ or both from ${\mathcal I}'$.

In the first case we have also
\begin{multline*}
\parallel_{r \in \roles(\mathcal{I}\seqOp\mathcal{I}')} \role{\pi({\mathcal
I},r),\Sigma_r}{r}\arro{\commLabel{o^?}{r_1}{v}{r_2}{x}} \parallel_{r \in
\roles(\mathcal{I}\seqOp\mathcal{I}')}  \role{P''_r,\Sigma_r'}{r}
\end{multline*} with $P_r=P''_r\seqOp\pi({\mathcal I}',r)$.
Thus, by inductive hypothesis, ${\mathcal I}
\arro{\commLabel{o^?}{r_1}{v}{r_2}{x}} {\mathcal I}''$ and $\upd(\parallel_{r
\in \roles{\mathcal{I}\seqOp\mathcal{I}'}} \role{P''_r,\Sigma_r}{r})$ is the
projection of ${\mathcal I}''$ with state $\Sigma$. Hence, we have that
${\mathcal I}\seqOp{\mathcal I'}
\arro{\commLabel{o^?}{r_1}{v}{r_2}{x}} {\mathcal I}''\seqOp{\mathcal I}'$. 

In the second case, thanks to Lemma~\ref{lemma:synchenabled}, we have that the
interaction is enabled. Thus, ${\mathcal I}$ has a transition with label
$\tick$ and ${\mathcal I'} \arro{
\commLabel{o^?}{r_1}{v}{r_2}{x}} {\mathcal I}''$.
Thanks to Lemma~\ref{lemma:tick} then both $\role{\pi({\mathcal I},r_1),
\Sigma_{r_1}}{r_1} $ and $\role{\pi({\mathcal I},r_2), \Sigma_{r_2}}{r_2}$
have a transition with label ${\tick}$. Thus, we have $\role{\pi({\mathcal
I}',r_1),\Sigma_{r_1}}{r_1} \arro{\coutLabel{o^?}{v}{r_2}:r_1}
\role{P_{r_1},\Sigma_{r_1}}{r_1}$, $\role{\pi({\mathcal
I}',r_2),\Sigma_{r_2}}{r_2} \arro{\cinpLabel{o^?}{x}{v}{r_1}:r_2}
\role{P_{r_2},\Sigma_{r_2}}{r_2}$ and \\ $\proj({\mathcal I}',\Sigma) \arro{
\commLabel{o^?}{r_1}{v}{r_2}{x}} \parallel_{r \in
\roles(\mathcal{I}')} \role{P_r,\Sigma_r}{r}$. The thesis follows by inductive
hypothesis. If $\eta$ uses an extended operation then the corresponding
\AIOC{} transition uses the corresponding basic operation.

For the other possibilities of $\eta$, only the process of one role changes.
Thus, the thesis follows by induction.

Note that in all the above cases, conditions on events follow by inductive
hypothesis.
\end{description}

\item[Case ${\mathcal I} \parOpI {\mathcal I}'$: ] from the definition of the
 projection function we have\\
 $\net = \parallel_{r \in \roles(\mathcal{I}\seqOp\mathcal{I}')} (\pi({\mathcal I},r) \mid
 \pi({\mathcal I}',r),\Sigma_r)_r$.
\begin{description}
\item[$\rightarrow$] If ${\mathcal I} \parOpI {\mathcal I}'$ can perform a
transition then one of its two components can perform the same transition and
the thesis follows by inductive hypothesis. Additional roles not occurring in
the term performing the transition are dealt with by function $\upd$.
\item[$\leftarrow$] We have a case analysis on $\eta$.  If
$\eta=\commLabel{o^?}{r_1}{v}{r_2}{x}$ then an input and an output on the same
operation are enabled. Thanks to Lemma~\ref{lemma:executematching} they have
the same global index. Thus they are from the same component and the thesis
follows by inductive hypothesis. For the other possibilities of $\eta$, only
the process of one role changes. The thesis follows by induction.  In all the
cases, roles not occurring in the term performing the transition are dealt
with by function $\upd$.
\end{description}

\item[Case $\ifthenKey{b \at r}{\mathcal{I}}{\mathcal{I}'}{n}$:] from the
definition of projection
\begin{multline*}
\net =  \parallel_{s \in \roles(\mathcal{I}) \cup
\roles(\mathcal{I}')\smallsetminus \{r\}} (\cinp{o^*_n}{x_n}{r}\seqOp\\
\ifthen{x_n}{\pi(\mathcal{I},s)}{\pi(\mathcal{I}',s)},\Sigma_s)_s \parallel \\
(\code{if} \; b \; \{ (\prod_{r' \in \roles(\mathcal{I}) \cup
\roles(\mathcal{I}') \smallsetminus
\{r\}} \cout{o^*_n}{\mathit{true}}{r'}) \seqOp \pi(\mathcal{I},r)\}\\
\code{else} \;  \{(\prod_{r' \in
\roles(\mathcal{I}) \cup \roles(\mathcal{I}')
\smallsetminus \{r\}} \cout{o^*_n}{\mathit{false}}{r'})\seqOp
\pi(\mathcal{I}',r)\}, \Sigma_r)_r
\end{multline*} Let us consider the case when the guard is true (the other
one is analogous).
\begin{description}
\item[$\rightarrow$] The only possible transition from the \AIOC{} is
$\ifthenKey{b \at r}{\mathcal{I}}{\mathcal{I}'}{n} \arro{\tau} \mathcal{I}$.
The \APOC{} can match this transition by reducing to
\begin{multline*}
\net' = \parallel_{s \in \roles(\mathcal{I}) \cup
\roles(\mathcal{I}')\smallsetminus \{r\}} (\cinp{o^*_n}{x_n}{r}\seqOp\\
\ifthen{x_n}{\pi(\mathcal{I},s)}{\pi(\mathcal{I}',s)},\Sigma_s)_s \parallel\\
(\prod_{r' \in \roles(\mathcal{I}) \cup \roles(\mathcal{I}') \smallsetminus
\{r\}} \cout{o^*_n}{\mathit{true}}{r'} \seqOp \pi(\mathcal{I},r),\Sigma_r)_r
\end{multline*} By applying function $\upd$ we get
\begin{multline*}
\upd(\net') = \parallel_{s \in \roles(\mathcal{I}) \cup
\roles(\mathcal{I}')\smallsetminus \{r\}}
\role{\pi(\mathcal{I},s),\Sigma_s}{s} \parallel
\role{\pi(\mathcal{I},r),\Sigma_r}{r}
\end{multline*} Concerning events, at the \AIOC{} level events corresponding
to the guard and to the non-chosen branch are removed. The same holds at the
\APOC{} level, thus conditions on the remaining events are inherited. This
concludes the proof.
\item[$\leftarrow$] The only possible transition from the \APOC{} is the
evaluation of the guard from the coordinator. This reduces $\net$ to $\net'$
above and the thesis follows from the same reasoning.
\end{description}

\item[Case $\whileKey{b \at r}{\mathcal{I}}{n}$:] from the definition of
projection {\footnotesize
\begin{multline*}
\net =  \parallel_{s \in \roles(\mathcal{I})\smallsetminus \{r\}}
(\cinp{o^*}{x_n}{r}\seqOp \\
\while{x_n}{\pi(\mathcal{I},s)\seqOp
\cout{o^*_n}{\texttt{ok}}{r}\seqOp\cinp{o^*_n}{x_n}{r}},\Sigma_s)_s
\parallel\\ (\code{while} \; b \; \{\prod_{r' \in \roles(\mathcal{I})
\smallsetminus
\{r\}} \cout{o^*_n}{\mathit{true}}{r'}\seqOp\pi(\mathcal{I},r)\seqOp\\
\prod_{r' \in \roles(\mathcal{I}) \smallsetminus \{r\}} \cinp{o^*_n}{\_}{r'}\}
\seqOp \\\prod_{r'\in \roles(\mathcal{I}) \smallsetminus \{r\}}
\cout{o^*_n}{\mathit{false}}{r'}, \Sigma_r)_r
\end{multline*}}
\begin{description}
\item[$\rightarrow$] Let us consider the case when the guard is true. The
only possible transition from the \AIOC{} is $\whileKey{b \at
r}{\mathcal{I}}{n} \arro{\tau}
\mathcal{I}\seqOp\whileKey{b \at r}{\mathcal{I}}{n}$. The \APOC{} can  match
this transition by reducing to {\footnotesize
\begin{multline*}
\net ' = \parallel_{s \in \roles(\mathcal{I})\smallsetminus \{r\}}
(\cinp{o^*}{x_n}{r}\seqOp\\
\while{x_n}{\pi(\mathcal{I},s)\seqOp \cout{o^*_n}{\texttt{ok}}{r}\seqOp
\cinp{o^*_n}{x_n}{r}},\Sigma_s)_s \parallel\\
(\prod_{r' \in \roles(\mathcal{I}) \smallsetminus \{r\}}
\cout{o^*_n}{\mathit{true}}{r'}\seqOp\pi(\mathcal{I},r)\seqOp\\
 \prod_{r' \in \roles(\mathcal{I}) \smallsetminus \{r\}}
 \cinp{o^*_n}{\_}{r'}\seqOp\\
 \code{while} \; b \; \{\prod_{r' \in \roles(\mathcal{I}) \smallsetminus
 \{r\}}
 \cout{o^*_n}{\mathit{true}}{r'}\seqOp\pi(\mathcal{I},r)\seqOp\\
 \prod_{r' \in \roles(\mathcal{I}) \smallsetminus \{r\}}
 \cinp{o^*_n}{\_}{r'}\}\seqOp\\
 \prod_{r'\in \roles(\mathcal{I}) \smallsetminus \{r\}}
 \cout{o^*_n}{\mathit{false}}{r'}, \Sigma_r)_r
 \end{multline*} } By applying function $\upd$ we get {\footnotesize
\begin{multline*}
\upd(\net') =  \parallel_{s \in \roles(\mathcal{I})\smallsetminus \{r\}}
(\pi(\mathcal{I},s)\seqOp
\cinp{o^*_n}{x_n}{r}\seqOp\\
\while{x_n}{\pi(\mathcal{I},s)\seqOp
\cout{o^*_n}{\texttt{ok}}{r}\seqOp\cinp{o^*_n}{x_n}{r}},\Sigma_s)_s \parallel
\\
(\pi(\mathcal{I},r)\seqOp
 \code{while} \; b \; \{\prod_{r' \in \roles(\mathcal{I}) \smallsetminus
 \{r\}}
 \cout{o^*_n}{\mathit{true}}{r'}\seqOp\pi(\mathcal{I},r)\seqOp\\
  \prod_{r' \in \roles(\mathcal{I}) \smallsetminus \{r\}}
  \cinp{o^*_n}{\_}{r'}\}\seqOp\\
\prod_{r'\in \roles(\mathcal{I}) \smallsetminus \{r\}}
\cout{o^*_n}{\mathit{false}}{r'}, \Sigma_r)_r
 \end{multline*}}
 which is exactly the projection  of $ \mathcal{I}\seqOp\whileKey{b \at
 r}{\mathcal{I}}{n}$.
 
As far as events are concerned, in $\prop(\net')$ we have all the needed
events since, in particular, we have already done the unfolding of the while
in all the roles. Concerning the ordering, at the \AIOC{} level, we have two
kinds of causal dependencies: (1) events in the unfolded process precede the
guard event; (2) the guard event precedes the events in the body.  The first
kind of causal dependency is matched at the \APOC{} level thanks to the
auxiliary synchronisations that close the unfolded body (which are not removed
by $\prop$) using synchronisation and sequentiality. The second kind of causal
dependency is matched thanks to the auxiliary synchronisations that start the
following iteration using  synchronisation, sequentiality and while.\\ The
case when the guard evaluates to \texttt{false} is simpler.
\item[$\leftarrow$] The only possible transition from the \APOC{} is the
evaluation of the guard from the coordinator. This reduces $\net$ to $\net'$
above and the thesis follows from the same reasoning.
\end{description}

\item[Case $n:\scope{l}{r}{\mathcal I}{\Delta}$:] from the definition of the
projection
\begin{multline*}
\net= \parallel_{s \in \roles(\mathcal{I}) \smallsetminus \{r\}}
\role{\psscope{n}{l}{r}{\pi({\mathcal I},s)},\Sigma_s}{s} \parallel \\
\pscope{n}{l}{r}{\pi({\mathcal I},r)}{\Delta}{\roles({\mathcal I})}
\end{multline*}

\begin{description}
\item[$\rightarrow$] The only possible transitions are obtained by applying
  rules \ruleName{Lead-\AdaptRule} or
  \ruleName{Lead-\NoAdaptRule} to the coordinator scope. Let us consider the
  first case.
\begin{multline*}
\net =  \parallel_{s \in \roles(\mathcal{I}) \smallsetminus \{r\}}
\role{\psscope{n}{l}{r}{\pi({\mathcal I},s)},\Sigma_s}{s} \parallel \\
\pscope{n}{l}{r}{\pi({\mathcal I},r)}{\Delta}{\roles({\mathcal I})} \\
 \arro{{\mathcal I}'} \parallel_{s \in \roles(\mathcal{I}) \cup \{r\}}
\role{P_s,\Sigma_s}{s} = \net'
 \end{multline*}
  
  For the coordinator we have:
\begin{multline*} P_r = \prod_{r_i \in \roles({\mathcal{I}}) \smallsetminus
\{r\}}\\
\cout{o^*_{n}}{\pi(\freshKey(\mathcal{I}',n),r_i)}{r_i}\seqOp \\
\pi(\freshKey(\mathcal{I}',n),r)\seqOp\\
 \prod_{r_i \in \roles({\mathcal{I}}) \smallsetminus \{r\}}
 \cinp{o^*_{n}}{\_}{r_i}
 \end{multline*}
For other roles $P_{r_i} = \psscope{n}{l}{r}{P}$. By applying the $\upd$
function we get:
\begin{multline*}
\upd(\net') = \pi(\freshKey(\mathcal{I}',n),r) \parallel \\ \parallel_{r_i \in
\roles(\mathcal{I}) \smallsetminus \{r\}} \pi(\freshKey(\mathcal{I}',n),r_i)
 \end{multline*}
This is exactly the projection of the \AIOC{} obtained after applying the rule
\ruleName{\AdaptRule}.
The conditions on events are inherited. Observe that the closing event of the
scope is replaced by events corresponding to the auxiliary interactions
closing the scope. This allows us to preserve the causality dependencies also
when the scope is inserted in a bigger context.
  
The case of rule \ruleName{Lead-\NoAdaptRule} is simpler.

\item[$\leftarrow$] The only possible transition from the \APOC{} is the one
of the coordinator of the scope checking whether to apply an update. This
reduces $\net$ to $\net'$ above and the thesis follows from the same
reasoning.
\end{description}

\end{description}
\end{proof}


\section{Proof of Corollary~\ref{cor1}}

Before proving Corollary~\ref{cor1}, we prove an auxiliary lemma.
\begin{lemma}\label{lemma:tickEnds}
For each initial, connected \AIOC{} ${\mathcal I}$, state $\Sigma$, and
set of updates $\rules$, if $\tuple{\Sigma,\rules,{\mathcal I}}
\arro{\tick} \tuple{\Sigma',\rules',{\mathcal I}'}$ then the only
transitions of $\tuple{\Sigma',\rules',{\mathcal I}'}$ have label
$\rules''$ for some $\rules''$.
\end{lemma}
\begin{proof}
The proof is by case analysis on the rules which can derive a
transition with label $\tick$. All the cases are easy.
\end{proof}

\setcounter{corollary}{0}
\begin{corollary}
For each initial, connected \AIOC{} ${\mathcal I}$, state $\Sigma$, and
set of updates $\rules$ the \APOC{} system $\tuple{\rules,\proj({\mathcal
    I},\Sigma)}$ is deadlock-free.
\end{corollary}
\begin{proof}
A \AIOC{} system $\tuple{\state, \rules, \mathcal{I}}$ is deadlock-free
if all its maximal finite internal traces have $\tick$ as label of the
last transition.  For each trace, the property can be proved by
induction on its length, and for each length by structural induction
on ${\mathcal I}$. The proof is based on the fact that ${\mathcal I}$
is initial. The induction considers a reinforced hypothesis, saying
also that $\tick$ never occurs before the end of the internal trace
and that all the steps, but the last one, lead to initial \AIOC{}s. We have
a case analysis on the top-level operator in ${\mathcal I}$. Note that in all the
cases at least a transition is derivable.
\begin{description}
\item[Case $\zero$:] not allowed since we assumed an initial \AIOC{}.
\item[Case $\one$:] trivial because by rule \ruleName{End} and Lemma \ref{lemma:tickEnds} its only 
internal trace is $\tick$.
\item[Case $\assign{x}{r}{e}$: ] the only applicable rule
  is \ruleName{Assign} that in one step leads to a $\one$ process. The thesis follows by 
  inductive hypothesis on the length of the trace.
\item[Case $\comm{o^?}{r_1}{e}{r_2}{x}$:] the only applicable rule
  is \ruleName{Interaction}, which leads to an assignment. Then
  the thesis follows by inductive hypothesis on the length of the
  trace.
\item[Case ${\mathcal I}\seqOp{\mathcal I}'$:] the first transition
  can be derived either by rule \ruleName{Sequence} or
  \ruleName{Seq-end}. In the first case the thesis follows by
  induction on the length of the trace. In the second case the trace
  coincides with a trace of ${\mathcal I}'$, and the thesis follows by
  structural induction.
\item[Case ${\mathcal I} \parOpI {\mathcal I}'$:] the first transition
  can be derived either by rule \ruleName{Parallel} or by rule
  \ruleName{Par-End}. In the first case the thesis follows by
  induction on the length of the trace. In the second case the thesis
  follows by Lemma~\ref{lemma:tickEnds}, since the label is
  $\tick$.
\item[Case $\ifthen{b \at r}{\mathcal{I}}{\mathcal{I}'}$:] the first transition can
  be derived using either rule \ruleName{If-then} or rule
  \ruleName{If-else}. In both the cases the thesis follows by
  induction on the length of the trace.
\item[Case $\while{b \at r}{\mathcal{I}}$:] the first transition can
  be derived using either rule \ruleName{While-unfold} or rule
  \ruleName{While-exit}. In both the cases the thesis follows by
  induction on the length of the trace.
\item[Case $\scope{l}{r}{\mathcal I}{\Delta}$:] the first rule applied
  is either \ruleName{Up} or \ruleName{NoUp}. In both the cases the
  thesis follows by induction on the length of the trace.
\end{description}
The weak internal traces of the \AIOC{} coincide with the weak internal
traces of the \APOC{} by Theorem \ref{teo:final}, thus the finite weak
internal traces end with $\tick$. The same holds for the finite
(strong) internal traces, since label $\tick$ is preserved when moving
between strong and weak traces, and no transition can be added after
the $\tick$ thanks to Lemma~\ref{lemma:tickEnds}.
\end{proof}

\section{Proof of Corollary~\ref{cor:race}}
\begin{corollary}[Race-freedom]
For each initial, connected \AIOC{} ${\mathcal I}$, state $\Sigma$, and
set of updates $\rules$, if $\tuple{\rules,\proj({\mathcal I},\Sigma)}
\arro{\mu_1}\cdots\arro{\mu_n} \tuple{\rules',\net}$, then in $\net$ two
outputs (resp.\ inputs) cannot interact with the same input (resp.\ output).
\end{corollary}
\begin{proof}
The result follows from Lemma~\ref{lemma:executematching}, which shows
that a \APOC{} transition always executes two matching events, since
for each input (resp.\ output) at most one matching output
(resp.\ input) exists.
\end{proof}

\section{Proof of Corollary~\ref{cor:orphan}}
\begin{corollary}[Orphan message-freedom]
 For each initial, connected \AIOC{} ${\mathcal I}$, state $\Sigma$, and
set of updates $\rules$, if $\tuple{\rules,\proj({\mathcal I},\Sigma)}
 \arro{\mu_1}\cdots\arro{\tick} \tuple{\rules',\net}$, then $\net$ contains no outputs.
\end{corollary}
\begin{proof}
The proof is by case analysis on the rules which can derive a transition with label $\tick$. All the 
cases are easy.
\end{proof}

}{
 \section{Rules for the \AIOC{} system semantics} \label{app:rulesAIOC}

Table \ref{table:ioclts-complete} presents the rules of the \AIOC{}
system semantics (symmetric rules for parallel
composition are omitted).

\begin{table*}[h]
{\footnotesize
\[
\begin{array}{c}
\mathrule{Interaction}{
\grass{e}_{\Sigma_{r_1}} = v
}{\tuple{\ambient, \comm{o^?}{r_1}{e}{r_2}{x}} 
\arro{\commLabel{o^?}{r_1}{v}{r_2}{x}} \tuple{\ambient,\assign{x}{r_2}{v}}}
\hfill\qquad
\mathrule{Sequence}{\tuple{\ambient,{\mathcal I}} \arro{\mu} \tuple{\ambient',{\mathcal I}'} \hfill 
\mu \neq \tick}{\tuple{\ambient,{\mathcal I}\seqOp{\mathcal J}} \arro{\mu} 
\tuple{\ambient',{\mathcal 
I}'\seqOp{\mathcal J}}}
\\
\mathrule{Assign}{
\grass{e}_{\Sigma_r} = v
}{
\tuple{\state,\rules,{\assign{x}{r}{e}}}
\arro{\tau}
\tuple{\state\substLabel{v}{x}{r},\rules, \one}
}
\hfill
\mathrule{Seq-end}{\tuple{\ambient,{\mathcal I}} \arro{\tick} \tuple{\ambient,{\mathcal I}'} \hfill 
\tuple{\ambient,{\mathcal J}} \arro{\mu} \tuple{\ambient,{\mathcal J}'}}{\tuple{\ambient,{\mathcal 
I}\seqOp{\mathcal J}} \arro{\mu} \tuple{\ambient,{\mathcal J}'}}
\\
\mathrule{Parallel}{\tuple{\ambient,{\mathcal I}} \arro{\mu} \tuple{\ambient',{\mathcal I}'} \hfill 
\mu \neq \tick}{\tuple{\ambient,{\mathcal I}\parallel{\mathcal J}} \arro{\mu} 
\tuple{\ambient',{\mathcal I}'\parallel{\mathcal J}}}
\hfill
\mathrule{Par-end}{\tuple{\ambient,{\mathcal I}} \arro{\tick} \tuple{\ambient,{\mathcal I}'} \hfill 
\tuple{\ambient,{\mathcal J}} \arro{\tick} \tuple{\ambient,{\mathcal 
J}'}}{\tuple{\ambient,{\mathcal I}\parallel{\mathcal J}} \arro{\tick} \tuple{\ambient,{\mathcal 
I}'\parallel{\mathcal J}'}}
\\
\mathrule{If-then}{
\grass{b}_{\Sigma_r} = \tt{true}}
{\tuple{\ambient,\ifthen{b \at r}{\mathcal{I}}{\mathcal{I}'}} 
\arro{\tau} \tuple{\ambient,{\mathcal I}}} \hfill
\mathrule{If-else}{
\grass{b}_{\Sigma_r} = 
\tt{false}}{\tuple{\ambient,\ifthen{b \at r}{\mathcal{I}}{\mathcal{I}'}} 
\arro{\tau} \tuple{\ambient,{\mathcal I}'}}
\\
\mathrule{While-unfold}{
\grass{b}_{\Sigma_r} = 
\tt{true}}{\tuple{\ambient,\while{b \at r}{\mathcal{I}}} \arro{\tau} 
\tuple{\ambient,{\mathcal I} \seqOp \while{b \at r}{\mathcal{I}}}} \hfill
\mathrule{While-exit}{
\grass{b}_{\Sigma_r} = 
\tt{false}}{\tuple{\ambient,\while{b \at r}{\mathcal{I}}} \arro{\tau} \tuple{\ambient,\one}} 
\\
\mathrule{\AdaptRule}{ 
\roles({\mathcal{I}'}) \subseteq \roles({\mathcal{I})} \quad \mathcal{I}'
\in \rules
\quad \mathcal{I}' \textrm{ connected}
}{\tuple{\ambient,\scope{l}{r}{\mathcal 
I}{\Delta}} \arro{{\mathcal I}'} \tuple{\ambient, {\mathcal 
I}'}}
\hfill
\mathax{\NoAdaptRule}{\tuple{\ambient,\scope{l}{r}{\mathcal I}{\Delta}}
\arro{\texttt{\NoAdaptLabel}} \tuple{\ambient,{\mathcal I}}}
\\
\mathax{End}{\tuple{\ambient,\one} \arro{\tick} \tuple{\ambient,\zero}}
\hfill
\mathax{Change-Updates}{\tuple{\state,\rules,{\mathcal I}} \arro{\rules'} 
\tuple{\state,\rules',{\mathcal I}}}
\end{array}
\]
}
\caption{\AIOC{} system semantics.}
\label{table:ioclts-complete}
\end{table*}
We describe here the rules whose descriptions were omitted in Table
\ref{table:ioclts} for space reasons.

Rule \ruleName{End}  terminates the execution of an empty process.
Rule
\ruleName{Sequence} executes a step in the first process of a sequential
composition, while rule \ruleName{Seq-end} acknowledges the termination of the
first process, starting the second one.
Rule \ruleName{Parallel} allows a process in a parallel composition to
compute, while rule \ruleName{Par-end} synchronises the termination of two
parallel processes.
Rules \ruleName{If-then} and \ruleName{If-else} evaluate the boolean guard of
a conditional, selecting the then and the else branch, respectively.
Rules \ruleName{While-unfold} and \ruleName{While-exit} correspond
respectively to the unfolding of a while when its condition is satisfied and
to its termination otherwise.

\section{Rules for the \APOC{} role semantics}
\label{app:rulesAPOC}

Table \ref{table:apoc-proc-complete} presents the rules of the \APOC{}
role semantics (symmetric rules for parallel composition are omitted).

\begin{table*}[h]
{\footnotesize
\[
\begin{array}{c}
\mathax{One}{\role{\one,\Gamma}{r} \arro{\tick}
\role{\zero,\Gamma}{r}}
\hfill\qquad
\mathax{Out-\AdaptRule}{\role{\cout{o^?}{X}{r'},\Gamma}{r} \arro{\coutLabel{o^?}{X}{r'}:r} 
\role{\one,\Gamma}{r}}
\hfill\qquad
\mathrule{Assign}{\grass{e}_\Gamma = v}{\role{x = e,\Gamma}{r} \arro{\tau} 
\role{\one,\Gamma\substAPOC{v}{x}}{r}}
\\
\mathax{In}{\role{\cinp{o^?}{x}{r'},\Gamma}{r} \arro{\cinpLabel{o^?}{x}{v}{r'}:r}\role{x = 
v,\Gamma}{r}}\hfill
\mathrule{Out}{\grass{e}_\Gamma = v}{\role{\cout{o^?}{e}{r'},\Gamma}{r} 
\arro{\coutLabel{o^?}{v}{r'}:r} \role{\one,\Gamma}{r}}
\\
\mathrule{Sequence}{\role{P,\Gamma}{r} \arro{\delta} \role{P',\Gamma'}{r} \quad \delta \neq 
\tick}{\role{P \seqOp Q,\Gamma}{r} \arro{\delta} \role{P' \seqOp Q,\Gamma'}{r}}
\hfill
\mathrule{Seq-end}{\role{P,\Gamma}{r} \arro{\tick} \role{P',\Gamma}{r} \quad \role{Q,\Gamma}{r} 
\arro{\delta} \role{Q',\Gamma'}{r}}{\role{P \seqOp Q,\Gamma}{r} \arro{\delta} \role{Q',\Gamma'}{r}}
\\
\mathrule{Parallel}{\role{P,\Gamma}{r} \arro{\delta} \role{P',\Gamma'}{r} \quad \delta \neq 
\tick}{\role{P \mid Q,\Gamma}{r} \arro{\delta} \role{P' \mid Q,\Gamma'}{r}}
\hfill
\mathrule{Par-end}{\role{P,\Gamma}{r} \arro{\tick} \role{P',\Gamma}{r} \quad \role{Q,\Gamma}{r} 
\arro{\tick} \role{Q',\Gamma}{r}}{\role{P \mid Q,\Gamma}{r} \arro{\tick} \role{P' \mid 
Q',\Gamma}{r}} 
\\
\mathrule{If-then}{\grass{b}_\Gamma = \tt{true}}{\role{\ifthen{b}{P}{P'},\Gamma}{r} 
\arro{\tau} \role{P,\Gamma}{r}}
\hfill
\mathrule{If-else}{\grass{b}_\Gamma = \tt{false}}{\role{\ifthen{b}{P}{P'},\Gamma}{r}
\arro{\tau} \role{P',\Gamma}{r}}
\\
\mathrule{While-unfold}{\grass{b}_\Gamma = \tt{true}}{\role{\while{b}{P},\Gamma}{r} \arro{\tau} 
\role{P \seqOp \while{e}{P},\Gamma}{r}}
\hfill
\mathrule{While-exit}{\grass{b}_\Gamma = \tt{false}}{\role{\while{b}{P},\Gamma}{r} \arro{\tau} 
\role{\one,\Gamma}{r}}
\\
\mathrule{Lead-\AdaptRule}{
\mathcal{I}' = 
\freshKey({\mathcal I}, n)
\qquad
\roles(\mathcal{I}') \subseteq S
}
{
\begin{array}{l}
\role{\pscope{n}{l}{r}{P}{\Delta}{S},\Gamma}{r} \arro{{\mathcal I}} 
\\ \quad \quad
\role{\prod_{r_i \in S \setminus \{r\}}
\cout{o^*_{n}}{\pi(\mathcal{I}',r_i)}{r_i}\seqOp
\pi(\mathcal{I}',r)\seqOp
\prod_{r_i \in S \setminus \{r\}} \cinp{o^*_{n}}{\_}{r_i},\Gamma}{r}\\
\end{array}
}
\\[1mm]
\mathax{Lead-\NoAdaptRule}{
\begin{array}{l}
\role{\pscope{n}{l}{r}{P}{\Delta}{S},\Gamma}{r}
\arro{\texttt{\NoAdaptLabel}} 
\\ \quad
\role{
\prod_{r_i \in S \setminus \{r\}} 
\cout{o^*_{n}}{\mbox{no}}{r_i}\seqOp P \seqOp \prod_{r_i \in S \setminus \{r\}} 
\cinp{o^*_{n}}{\_}{r_i},\Gamma}{r}\\
\end{array}
} 
\\[1mm]
\mathax{\AdaptRule}{\role{\psscope{n}{l}{r'}{P},\Gamma}{r} 
\arro{\cinpLabel{o^*_{n}}{\_}{P'}{r'}}
\role{P' \seqOp
\cout{o^*_{n}}{\texttt{ok}}{r'},\Gamma}{r}}
\\[1mm]
\mathax{\NoAdaptRule}{\role{\psscope{n}{l}{r'}{P},\Gamma}{r} 
\arro{\cinpLabel{o^*_{n}}{\_}{\mbox{no}}{r'}}
\role{P \seqOp \cout{o^*_{n}}{\texttt{ok}}{r'},\Gamma}{r} }
\end{array}
\]
}
\caption{\APOC{} role semantics.}\label{table:apoc-proc-complete}
\end{table*}
We describe here the rules whose descriptions were omitted in Table
\ref{table:apoc-proc} for space reasons.

Rules \ruleName{Sequence} executes a step in the first process of a sequential composition, while rule \ruleName{Seq-end} acknowledges the termination of the first process, starting the second one.
Rules \ruleName{Parallel} allows a process in a parallel composition
to execute, while rule \ruleName{Par-end} synchronises the termination
of two parallel processes.
Rules \ruleName{If-then} and \ruleName{If-else} execute the then or the else
branch in a conditional, respectively.
Rules \ruleName{While-unfold} and \ruleName{While-exit} model the unfolding or
the termination of a loop.

}

\end{document}